\documentclass{article}

\usepackage{arxiv}
\usepackage{graphicx}
\usepackage{natbib}
\usepackage{amsmath,amssymb,graphicx}
\usepackage[english]{babel}
\usepackage[latin1]{inputenc}

\usepackage{url}
\usepackage{pdflscape}

\usepackage{caption}
\usepackage{subcaption}    
\usepackage[colorlinks,pdftex]{hyperref}
\usepackage[table]{xcolor} 
\usepackage{arydshln}

\usepackage[ruled,vlined]{algorithm2e}
\usepackage{algcompatible}

\usepackage{multirow}

\usepackage{multicol}
\usepackage{array}
\newcolumntype{H}{>{\setbox0=\hbox\bgroup}c<{\egroup}@{}}
\usepackage{colortbl}
\usepackage{tabularx}
\usepackage{pdflscape}

\usepackage{pgf,tikz}
\usetikzlibrary{shapes,arrows}
\usepackage[title]{appendix}  

\usepackage[miktex]{gnuplottex}

\definecolor{Gray}{gray}{0.85}
\definecolor{Gray_0}{gray}{0.95}
\definecolor{Gray_1}{gray}{0.60}
\definecolor{Cyan}{rgb}{0,1,1}
\definecolor{Blue}{rgb}{0,0.1,1}
\definecolor{Red}{rgb}{1,0.7,0.7}

\newcommand{\rz}{{\mathbb{R}}}

\algnewcommand\algorithmicreturn{\textbf{return}}
\algnewcommand\RETURN{\State \algorithmicreturn}%

\newcommand{\maxim}{\mathop{\mathrm{maximize}}}

\newcommand{\NE}{NE}
\newcommand{\CE}{CE}
\newcommand{\IPG}{IPG}
\newcommand{\RIPG}{RIPG}
\newcommand{\PNS}{PNS}
\newcommand{\SGM}{SGM}
\newcommand{\modSGM}{m-SGM}
\newcommand{\xqed}{~$\Box$}
\newcommand{\ie}{\emph{i.e.}}
\newcommand{\etal}{\emph{et~al.}}
\newcommand{\eg}{\emph{e.g.}}
\newcommand{\supp}{\textrm{supp}}
\newcommand{\conv}{\textrm{conv}}

\newtheorem{axiom}{Axiom}

\newtheorem{corollary}[axiom]{Corollary}
\newtheorem{definition}[axiom]{Definition}
\newtheorem{example}{Example}
\newtheorem{lemma}[axiom]{Lemma}

\newtheorem{theorem}[axiom]{Theorem}

\newcommand{\pulley}{{\sc Subset-Sum-Interval}}
\newcommand{\boxxx}[1]
 {\begin{center}\fbox{\begin{minipage}{12.00cm}#1\smallskip\end{minipage}}\end{center}}

\begin{document}

\title{Computing Nash equilibria for\\ integer programming games\thanks{M. Carvalho  is thankful for the support of the Institut de valorisation des donn\'ees (IVADO) and Fonds de recherche du Qu\'ebec (FRQ) through the FRQ- IVADO	Research Chair, NSERC grant 2019-04557,  and the Portuguese Foundation for Science and Technology (FCT) through a PhD grant number SFRH/BD/79201/2011.}
}

\author{
 Margarida Carvalho
  \\
  CIRRELT and D\'epartement d'Informatique et de Recherche Op\'erationnelle\\
   Universit\'e de Montr\'eal\\
  \texttt{carvalho@iro.umontreal.ca} \\
  \And
Andrea Lodi \\
 Canada Excellence Research Chair in Data Science for Real-Time Decision-Making\\
 Polytechnique de Montr\'eal\\
 \texttt{andrea.lodi@polymtl.ca} \\
 \And
Jo\~ao Pedro Pedroso\\
INESC TEC and Departamento de Ci\^encia de Computadores \\
Faculdade de  Ci{\^e}ncias Universidade do Porto\\
 \texttt{jpp@fc.up.pt} \\
}

\maketitle

\begin{abstract}
The recently defined class of integer programming games ({\IPG}) models situations where multiple self-interested decision makers interact, with their strategy sets represented by a finite set of linear constraints together with integer requirements. Many real-world problems can suitably be fit in this class, and hence anticipating {\IPG} outcomes  is of crucial value for policy makers and regulators. Nash equilibria have been widely accepted as the solution concept of a game. Consequently, their computation provides a reasonable prediction of games outcome.

In this paper, we start by showing the computational complexity of deciding the existence of a Nash equilibrium for an {\IPG}. Then, using sufficient conditions for their existence, we develop two general algorithmic approaches that are guaranteed to approximate an equilibrium under mild conditions. We also showcase how our methodology can be changed to determine other equilibria definitions. The performance of our methods is analysed through computational experiments 
in a knapsack game, a competitive lot-sizing game and a kidney exchange game. To the best of our knowledge, this is the first time that equilibria computation methods for general integer programming games have been designed and computationally tested. 
\end{abstract}

\section{Problem statement and background}\label{sec:problemStat}

 We start by defining some general notation. If $C^i$ is a set for each $i \in M$, then we denote their Cartesian product as $C= \prod_{i \in M} C^i$. The operator $(\cdot)^{-i}$ is used to denote $(\cdot)$ for all $j\in M \setminus \lbrace i \rbrace$; \eg, if $M=\{1,2,3\}$, for a vector $x=(x^1,x^2,x^3)$, we have $x^{-1}=(x^2,x^3)$ and, for a Cartesian product of sets $C= C^1 \times C^2 \times C^3$, we have $C^{-2}=C^1\times C^3$. If $C$ is a set, we use the notation $\Delta(C)$ to represent the space of Borel probability measures over $C$. If  $C= \prod_{i \in M} C^i$, the associated set of independent probability distributions is denoted by $\Delta =\prod_{i \in M} \Delta(C^i) $. For the latter notations note that $\Delta\subseteq \Delta(C)  $.

\paragraph{Integer programming games ({\IPG}s). } Following the seminal work in \citep{Koppe:2011}, \cite{carvalhoPhD, Carvalho_APDIO} defined \emph{integer programming games}. An {\IPG}  is   a game with a finite set of \emph{players} $M=\lbrace 1,2,\ldots,m \rbrace$ such that for  each player $p \in M$, the  \emph{set of strategies} is given by $$ X^p=  \lbrace x^p: A^p x^p \leq b^p, \ \ x^p_i \in \mathbb{N} \textrm{ for } i=1, \ldots, B_p \rbrace,$$
where $A^p$ is an $r_p \times n_p$ rational matrix (where $n_p \geq B_p$) and $b^p$ is a rational column vector of dimension $r_p$. An $x^p \in X^p$ is called a (\emph{pure}) strategy of player $p$.  Each player $p$ has a continuous payoff function $\Pi^p: X \rightarrow \rz$ that can be evaluated in polynomial time. Note that accordingly with our notation, $X$ is the set $\prod_{i \in M} X^p$ which corresponds to all possible game outcomes, \ie, all possible combinations of players' strategies. An $x \in X$ is called a \emph{profile of strategies}.

An {\IPG} is a \textit{non-cooperative complete information game}, \ie, players are self-interested and have full information of each other's
payoffs and strategies. We restrict our focus to the \textit{simultaneous} case, \ie, players select their strategies simultaneously.

\paragraph{Extensions for mixed strategies. }
Under the simultaneous play assumption, as motivated later, players may consider to randomize among their pure strategies. Hence, for a player $p \in M$, it is of interest to consider her set of \emph{mixed strategies} $\Delta(X^p)$. For a player $p$'s mixed strategy $\sigma^p \in \Delta(X^p)$, its \emph{support} is defined   as $\supp(\sigma^p) = \lbrace x^p \in X^p: \quad \sigma^p(x^p)>0 \rbrace$,  \ie, the set of player $p$'s strategies played with strictly positive probability.  A  $\sigma \in \Delta = \prod_{i \in M} \Delta(X^p)$ is called a \emph{mixed profile of strategies}, and  if $\vert \supp(\sigma^p) \vert=1$ holds for all $p \in M$, $\sigma$ is called a \textit{pure profile of strategies}. For the sake of simplicity, whenever the context makes it clear, we use the term  (strategy) profile to refer to a pure profile.  A player $p$'s expected payoff for a profile of strategies $\sigma \in \Delta$ is
\begin{equation}
	\Pi^p(\sigma) = \int_{X} \Pi^p(x^p,x^{-p}) d \sigma.
	\label{expected_payoff}
\end{equation}
The same definition of expected payoff covers joint probability distributions $\tau \in \Delta(X) \supseteq \Delta$ with $\tau(x)$ representing the probability assigned to profile $x$. Similarly,  the support of $\tau \in \Delta(X) $ is defined as $\supp(\tau)=\{ x \in X: \tau(x)>0\}$.

\paragraph{Solution concepts. } Ideally, each player would like to ensure that her \emph{best response}, also designated by best reaction, is selected given the opponents' strategy $\sigma^{-p} \in \Delta^{-p}$. In other words, each player $p$ solves
\begin{equation}
	\maxim_{x^p \in X^p}  \quad  \Pi^p(x^p,\sigma^{-p}),
	\label{GeneralProblem}%
\end{equation}
where for sake of simplicity, $(x^p,\sigma^{-p})$ denotes the profile of strategies in which the pure strategy $x^p$ is played with probability 1 by player $p$ and the remaining players behave accordingly with $\sigma^{-p}$. Note that a mixed strategy for a player $p$ is simply a convex combination of her pure strategies. Thus, when computing best responses it is sufficient to restrict ourselves to pure strategies as done in Problem~\eqref{GeneralProblem}. An {\IPG} is completely defined given  Problem~\eqref{GeneralProblem} for all players. Therefore, for now on, we will use them to represent an {\IPG}. 

We now introduce the most broadly accepted concept of solution for a game.

\begin{definition}
A \emph{Nash equilibrium} (\NE) \citep{Nash1950} is a profile of strategies $\sigma  \in \Delta$ such that
\begin{equation}
	\Pi^p(\sigma)  \geq \Pi^p(x^p, \sigma^{-p}),  \qquad \forall p \in M \qquad \forall x^p \in X^p.
	\label{NE_definition}
\end{equation}
The profile  $\sigma$ is called a mixed Nash equilibrium, and if it is a pure profile of strategies, it is also called a pure Nash equilibrium.
\end{definition}
It is now easy to verify if a given $\sigma  \in \Delta$ is a {\NE} by computing each player $p \in M$ best response to $\sigma^{-p}$ (\ie, by solving Problem~\eqref{GeneralProblem}) and confirming that she cannot increase her payoff more than $\Pi^p(\sigma)$, \ie, inequalities~\eqref{NE_definition} are not violated. In other words, in an {\NE}, no player has incentive to unilaterally deviate from it.

The following two definitions are relaxations of the concept of Nash equilibrium which are of interest to this work.

\begin{definition}
An $\varepsilon$-equilibrium
($\varepsilon \geq 0$) is a profile of strategies $\sigma  \in \Delta$ such that
\begin{equation}
	\Pi^p(\sigma) +\varepsilon  \geq \Pi^p(x^p, \sigma^{-p}),  \qquad \forall p \in M \qquad \forall x^p \in X^p .
	\label{NE_definition_epsilon}
\end{equation}
\end{definition}

\begin{definition}
 A joint probability distribution   $\tau \in \Delta(X)$ is a correlated equilibrium (\CE)~\citep{Aumann1974,Aumann1987} if
\begin{equation}
\int_{X^{-p} \cup \{\bar{x}^p\}} \Pi^p(\bar{x}^p,x^{-p}) \ \ d \tau \geq  \int_{X^{-p} \cup \{\bar{x}^p\}} \Pi^p(\hat{x}^p,x^{-p}) \ \ d \tau  \ \ \forall p \in M, \forall \bar{x}^p,\hat{x}^p \in  X^p.  
	\label{Correlated_equilibria}%
\end{equation}
\end{definition}

In an  $\varepsilon$-equilibrium, no player can unilaterally deviate from it and increase her payoff by more than $\epsilon$. In a correlated equilibrium  a joint probability distribution is considered instead of an independent one for each player.   Correlated equilibria can be interpreted as a third party signaling the players on how they play such that deviating from that recommendation does not increase their payoffs (Inequalities~\eqref{Correlated_equilibria}).  We remark that the set of correlated equilibria contains the set of Nash equilibria.

\paragraph{Preliminary results.} The goal of this work is to compute equilibria for {\IPG}s. However, the fact that players can have continuous variables means that their strategy sets can be uncountable. Thus, the support of an equilibrium $\sigma$ can also be uncountable. Next, we state a set of sufficient conditions that enable us to  restrict to equilibria with finite support.


\begin{definition}
	Player $p$'s payoff function is called \emph{separable} if
\begin{equation}
	\displaystyle \Pi^p(x)=\sum_{j_1=1}^{k_1} \ldots \sum_{j_m=1}^{k_m}  a^p_{j_1 \ldots j_m} f^1_{j_1}(x^1) \ldots f^m_{j_m}(x^m), 
	\label{ObjectivePlayer_separable}
\end{equation}
where  $a^p_{j_1 \ldots j_m} \in \mathbb{R}$ and the $f^p_j$ are real-valued continuous functions.
\end{definition}
An {\IPG} where all players' payoff functions are separable (\ie, take the form~\eqref{ObjectivePlayer_separable}) and strategy sets are nonempty and bounded is called \emph{separable}.

\begin{example}
Consider a 2-player game, $M=\{1,2\}$, with payoff functions 
\begin{align*}
\Pi^1(x) &= x^1_1 \cdot x_2^1 +  x^1_1 \cdot x_3^1 \cdot x^2_1\\
\Pi^2(x)&=x^1_2 \cdot x^1_3 \cdot x_1^2.
\end{align*}
Both players' payoffs are separable as they take the form~\eqref{ObjectivePlayer_separable}: $k_1=3$, $k_2=2$, $f^1_1= x^1_1 \cdot x_2^1$, $f^1_2= x^1_1 \cdot x_3^1$, $f^1_3=x^1_2 \cdot x^1_3$, $f^2_1=1$, $f^2_2=x^2_1$, $a^1_{11}=a^1_{22}=a^2_{32}=1$ and the remaining $a$ coefficients are zero.
	
\end{example}

In \cite{Carvalho_APDIO} the following useful results based on~\cite{Stein2008} were proven:
\begin{theorem}[\cite{Carvalho_APDIO}]
	Every {\IPG} such that $X^p$ is nonempty and bounded for all $p \in M$  has a Nash equilibrium.
	\label{them_existence}
\end{theorem}

\begin{theorem}[\cite{Carvalho_APDIO}]
	For any Nash equilibrium $\sigma$ of a separable {\IPG}, there is a Nash equilibrium $\bar{\sigma}$ such that each player $p$ mixes among at most $k_p+1$ pure strategies and $\Pi^p(\sigma) = \Pi^p(\bar{\sigma})$. 
	\label{lemma_finitelysupported}
\end{theorem}

Theorem~\ref{them_existence} ensures that under a mild condition on the players' sets of strategies, an {\IPG} has an \NE. Furthermore, if an {\IPG} is separable, any {\NE} can be converted in a payoff-equivalent {\NE}  with a finite support. Since a {\NE} is a {\CE}, any separable {\IPG} has a {\CE} with finite support. In this work, we will thus focus on equilibria with finite support. Consequently, for finitely-supported $\sigma \in \Delta$, player $p$'s expected payoff is
\begin{equation}
	\Pi^p(\sigma) = \sum_{x \in \supp(\sigma)} \Pi^p(x) \prod_{i \in M} \sigma^i(x^i),
	\label{expected_payoff_sigma_discrete}
\end{equation}
and for $\tau \in \Delta(X)$, it is
\begin{equation}
	\Pi^p(\tau) = \sum_{x \in \supp(\tau)} \Pi^p(x) \tau(x).
	\label{expected_payoff_tau_discrete}
\end{equation}

To end this section, we define \emph{potential games} for which the existence of pure {\NE} can be guaranteed. 

\begin{definition}
A game is \emph{potential}~\cite{Monderer1996124} if there is a real-valued function $\Phi: X \longrightarrow \mathbb{R}$ such that its value strictly increases  when a player switches to a strategy that strictly increases her payoff.
\end{definition}

\begin{lemma}[\cite{Monderer1996124}]
The maximum of a potential function for a game is a pure Nash equilibrium.
\label{lem:Monderer_Shapley}
\end{lemma}

\section{Contributions }

In~\cite{Carvalho_APDIO}, the authors discuss the existence of Nash equilibria for integer programming games. It is proven that deciding the existence of pure Nash equilibria for {\IPG}s is  $\Sigma^p_2$-complete and that even the existence of Nash equilibria is  $\Sigma^p_2$-complete. However, the latter proof seems incomplete in the ``proof of only if''. Thus, our first contribution is the presentation of a completely new and correct proof (reduction).

Our second and main contribution is in the development of a flexible framework to compute an {\NE} for {\IPG}s. Based on the theorems of the previous section, we are able to show that   our framework  \emph{(i)} is guaranteed to compute an {\NE} for {\IPG}s in which all the players' sets of strategies are nonempty and bounded lattice points, and \emph{(ii)} it is guaranteed to compute an $\varepsilon$-equilibrium for {\IPG} under some mild conditions that are expected to be satisfied in real-world games. Nevertheless, our framework is capable of processing any {\IPG}, although, it might fail to stop, \eg, if the input game has no equilibria.

Our framework requires game theory and mathematical optimization algorithms. In fact, it is an iterative approach integrating different components from both fields. Those components can be divided in optimization algorithms, search of {\NE}, and heuristics. Each of them offers the user the flexibility of selecting the algorithm/solver most appropriated for the {\IPG} at hand. However, for the search of {\NE} solver, we strongly advise the use of Porter-Nudelman-Shoham method due to its practical efficiency, simple implementation  and easy integration of heuristics.  We also show how to adapt our method to determine correlated equilibria. 

To conclude the paper, we evaluate our methodology's performance through computational experiments in three different integer programming games: the knapsack game, the lot sizing game and the kidney exchange game. Given that this is the first general-purpose algorithm for {\IPG}s, there is no other  method in the literature to which our experiments can be compared.

Our paper is structured as follows. Section~\ref{sec:related_literature} reviews the literature in algorithmic game theory for the computation of Nash equilibria. In Section~\ref{sec:complexity}, we classify the computational complexity of deciding the existence of {\NE} for {\IPG}s.  Section~\ref{sec:algorithm}  formalizes our framework, develops two methods to
compute $\varepsilon$-equilibria for {\IPG}s (\textit{approximated} \NE), providing specialized functions to speed up
the methods, and extensions to {\CE}. In Section~\ref{sec:Computational_investigation}, we introduce three relevant {\IPG}s, and validate our methods through computational experiments on these games.  Finally, we conclude and discuss further research directions in Section~\ref{sec:conclusion}.

\section{Related literature}\label{sec:related_literature}

There are important real-world games  (\eg, in electricity markets~\citep{Fampa2005}, production planning~\citep{Li2011535}, health-care~\citep{Carvalho2016}, where each player's payoff maximization subject to her set of feasible strategies is described  by  a mixed integer programming formulation as required in the definition of {\IPG}s. This motivates the importance of understanding the equilibria of {\IPG}s, as they indicate their likely outcome and thus, its impact to the participants (players) and to the society. Concretely, in the game examples mentioned, the players are companies and countries that provide services to the population. Hopefully, this competition will be sufficiently profitable to the players so that they can create jobs, invest in technological innovation, while providing high quality service to the population. Thus, the computation of equilibria can allow us to anticipate these games outcomes and serve policy makers in their task of guaranteeing social welfare.

Moreover, {\IPG}s contain the well-known class of \emph{finite} games~\citep{Carvalho_APDIO}, \ie, games with a finite number of strategies and players, and \emph{quasi-concave} games, \ie, game with convex strategies sets and quasi-concave payoffs. The existing tools and standard approaches for finite games and quasi-concave games are not directly applicable to general {\IPG}s. Additionally,  the previous literature on {IPG}s focuses in the particular structure of specific games. 

\paragraph{Pure Nash equilibria.} 
\cite{Kostreva199327} describes the first theoretical approach to compute pure {\NE} to {\IPG}s, where  integer variables are required to be binary. The binary requirement in a binary variable $x$ is relaxed by adding in the payoff a penalty $Px(1-x)$ where $P$ is a very large number. Then,  the Karush-Kuhn-Tucker (KKT)~\citep{Karush39,kuhn1951} conditions are applied to each player optimization problem and merged into a system of equations for which the set of solutions contains the set of pure equilibria. To find the solutions for that system of equations, the author recommends the use of a homotopy path following~\citep{zangwill1981pathways} or Gr\"obner basis~\citep{Cox:2007:IVA:1204670}. Additionally, it must be verified which of the system's solutions are equilibria\footnote{The KKT conditions applied to non-concave maximization problems are only necessary.}, which results in long computational times. \cite{GabrielSiddiqui2013} proposed an optimization model for which the optimal solution is a pure Nash equilibrium of a game that approximates an {\IPG} with concave payoffs. In that paper, integer requirements are relaxed, the players' concave optimization problems are transformed in constrained problems through the KKT conditions; then, the complementary conditions are also relaxed  but their deviation from zero is minimized. On the few experimental results presented, this approach leads to a pure Nash equilibrium for the original game. However, there is neither a theoretical nor computational evidence showing the applicability of these ideas to the general case.  \cite{Hemmecke2009} considered {\IPG}s with an additional feature: a player $p$'s set of feasible strategies depends on the opponents' strategies. The authors study (generalized) pure equilibria assuming that the player's decision variables are all  integer and   payoffs are monotonously decreasing in each variable.   \cite{Koppe:2011} were the pioneers to investigate the computation of all pure {\NE} to {\IPG}s where all the players' decision variables are integer and their payoffs are differences of piecewise-linear concave functions. In order to compute {\NE}, the authors use generating functions of integer points inside of polytopes. The application of K\"oppe \etal's results relies on computational implementations that are still in the preliminary stage, although theoretically the approach can be proven to run in polynomial time under restrictive conditions, such as a fixed number of players and a fixed  number of players' decision variables, to name a few.  More recently, \cite{Pia2017} concentrated on the computation of pure {\NE} for {\IPG}s where the strategy sets are given by totally unimodular constraint matrices. They identify the cases where such games are potential and pure equilibria can be computed in polynomial time, and showed some cases where computing pure equilibria is PLS-complete (Polynomial Local Search).

\paragraph{Mixed Nash equilibria.} \cite{Kwang2003} studied the computation of mixed equilibria for an {\IPG} in the context of the electric power market. There, the players' set of strategies is approximated through its discretization, resulting in a finite game to which there are general algorithms to compute {\NE}. Nevertheless, there is a trade-off between having a good discretized approximation and an efficient computation of {\NE}: the more strategies are contained in the discretization, the longer the time to compute a {\NE} will be. \cite{Stein2008,Steinthesis} restricted their attention to separable games, meaning that all their results hold for separable {\IPG}s. The authors were able to provide bounds on the cardinality of the {\NE} support and present a polynomial-time algorithm for computing $\varepsilon$-equilibria of two-player separable games with fixed strategy spaces and payoff functions satisfying the H\"older condition.

None of the approaches above tackles general {\IPG}s, failing to either consider mixed {\NE} or continuous and integer decision variables for the players.

\section{Computational complexity}\label{sec:complexity}

In what follows, we show that even in the simplest case, linear integer programming games with two players, the existence of Nash equilibria is a  $\Sigma^p_2$-complete problem.

\begin{theorem}
The problem of deciding if an {\IPG} has a Nash equilibrium is $\Sigma^p_2$-complete problem.
\end{theorem}
\begin{proof}
The proof that this decision problem belongs to  $\Sigma^p_2$ can be found in~\cite{Carvalho_APDIO}. It remains to show that it is  $\Sigma^p_2$-hard. We will reduce the following  $\Sigma^p_2$-complete probleme~\citep{EgWo2012} to it:

\boxxx{
\vspace{0.2cm}
\textbf{Problem: {\pulley}}

\vspace{0.2cm}
\textbf{INSTANCE } A sequence $q_1,q_2,\ldots,q_k$ of positive integers; two positive integers $R$ and $r$ with $r\leq k$.

\vspace{0.2cm}
\textbf{QUESTION }Does there exist an integer $S$ with $R\leq S<R+2^r$ such that none of the
subsets $I\subseteq\{1,\ldots,k\}$ satisfies $\sum_{i\in I}q_i=S$?
}

Our reduction starts from an instance of {\pulley}. We construct the following instance of {\IPG}:
\begin{itemize}
\item The game has two players, $M=\lbrace Z, W \rbrace$, with player $Z$ ($W$) controlling the decision vector $z$ ($w$).
\item Player $Z$ solves
\begin{subequations}
\begin{alignat}{4}
  \max_{z} &  \quad   \frac{1}{2}z_0+\sum_{i=1}^k q_i z_i + Qz(2w-z)\\
  s.t.  &\quad  \frac{1}{2}z_0+\sum_{i=1}^k q_i z_i \leq z\\
  & \quad z_0, z_1, \ldots, z_k \in \lbrace 0,1 \rbrace\\
  & \quad R \leq z \leq R+2^r-1, z \in \mathbb{N}.
\end{alignat}
\label{PlayerZ}%
\end{subequations}
where $Q=\sum_{i=1}^k q_i$. We add binary variables $y \in \lbrace 0,1\rbrace^r$ and we make $z =R + \sum_{i=0}^{r-1} 2^i y_i$. Note that $z^2= Rz+\sum_{i=0}^{r-1} 2^i y_iz$. Thus, we can replace $y_iz$ by a new variable $h_i$ and add the respective McCormick constraints~\cite{McCormick76}. In this way, we can equivalently linearize the previous problem:
\begin{subequations}
\begin{alignat}{4}
  \max_{z,y,h} &  \quad   \frac{1}{2}z_0+\sum_{i=1}^k q_i z_i +2 Qzw-QRz-\sum_{i=0}^{r-1} 2^i  h_i\\
  s.t.  &\quad  \frac{1}{2}z_0+\sum_{i=1}^k q_i z_i \leq z\\
  & \quad z_0, z_1, \ldots, z_k \in \lbrace 0,1 \rbrace\\
  & \quad R \leq z \leq R+2^r-1, z \in \mathbb{N}\\
  & \quad z = R+ \sum_{i=0}^{r-1} 2^i y_i\\
  & \quad y_0, y_1, \ldots, y_{r-1} \in \lbrace 0,1 \rbrace\\
  & \quad h_i \geq 0 &i=0,\ldots,r-1\\
  & \quad h_i \geq z+(R+2^r-1)(y_i -1) & i=0,\ldots,r-1 \\
  & \quad h_i \leq z+R(y_i-1) & i=0,\ldots,r-1\\
  & \quad h_i \leq (R+2^r-1)y_i & i=0,\ldots,r-1.
\end{alignat}
\label{PlayerZlin}%
\end{subequations}
For sake of simplicity of our reasoning, we consider the quadratic formulation (\ref{PlayerZ}). The linearization above serves the purpose of showing that the proof is valid even under linear payoff functions for the players. 
\item Player W solves 
\begin{subequations}
\begin{alignat}{4}
  \max_{w} &  \quad   (1-z_0)w_0\\
  s.t.  & \quad R \leq w \leq R+2^r-1 
  \\
  & w_0 \in \mathbb{R}.
\end{alignat}
\label{PlayerW}%
\end{subequations}
\end{itemize}

(Proof of if).  Assume that the {\pulley} instance has answer YES. Then, there is an $S$ such that $R\leq S<R+2^r$  and for all subsets $I\subseteq\{1,\ldots,k\}$,  $\sum_{i\in I}q_i\neq S$. Let player $W$ strategy be $w^*=S$ and $w_0^*=0$. Note that the term $Qz(2w-z)$ in player $Z$'s payoff is dominant and attains a maximum when $z$ is equal to $w$. Thus, we make $z^*=w^*=S$ and since $\sum_{i=1}^k q_i z_i$ is at most $S-1$, we also make $z_0^*=1$. Next, we choose $z_i^*$ such that the remaining payoff of player $Z$ is maximized. By construction, player $Z$ is selecting her best response to $(w^*,w_0^*)$. Since $z_0^*=1$, then player $W$ is also selecting an optimal strategy. Therefore, we can conclude that this is an equilibrium.

(Proof of only if). Assume that the {\pulley} instance has answer NO. Then, for all $S$ such that $R\leq S<R+2^r$, there is a subset $I\subseteq\{1,\ldots,k\}$ with $\sum_{i\in I}q_i=S$. In this case, player $Z$ will always make $z_0=0$ which gives incentive for player $W$ to choose $w_0$ as  large as possible. Since $w_0$ has no upper bound, there is no equilibrium for the game.
\end{proof}

\section{Algorithmic approach}\label{sec:algorithm}

As shown in the previous section, the problem of deciding the existence of {\NE} for {\IPG}s is complete for the second level of the polynomial hierarchy\footnote{ The second level of the polynomial hierarchy is  $\Sigma_2^p$.}, which is a class of problems believed to be hard to solve. In fact, even when an {\IPG} is guaranteed to have an {\NE}, it is unlikely that it can be determined in polynomial time. To provide evidence in this direction, the following definition is required.

\begin{definition}
A \emph{normal-form} game, also called strategic-form game, is a finite game whose description is given by a multidimensional payoff matrix for all possible pure strategy profiles.
\end{definition}

Any normal-form game can be equivalently reformulated as an {\IPG} in polynomial time~\citep{Carvalho_APDIO}: essentially, for each player, one just needs to associate a binary variable for each of her pure strategies and a constraint enforcing that only one variable takes value 1, \ie, only one pure strategy is played. \cite{XiChen2006} proved that computing an {\NE} for a normal-form game, even with only two players, is PPAD-complete\footnote{PPAD stands for Polynomial Parity Arguments on Directed graphs.}. In simple words, for a PPAD-complete problem it is known that a solution exists; however the proof of solution existence is non-constructive and it is believed to be ``hard" to compute it. The result in~\cite{XiChen2006} together with Theorem~\ref{them_existence} and the fact that finite games (and thus, normal-form games) are separable (see  \cite{Carvalho_APDIO}) leads to:
\begin{lemma}
	The problem of computing an {\NE} to an {\IPG} with non-empty bounded strategy sets is PPAD-hard, even for  separable {\IPG} with only binary variables.
\end{lemma}

Despite of this theoretical intractability evidence, in what follows, we leverage on the power of mixed integer programming solvers and practical {\NE} search approaches to build an efficient framework for the computation of equilibria to {\IPG}s in practice. In the remainder of the paper, we focus on separable {\IPG}s since their set of {\NE} can be characterized by finitely-supported equilibria (Theorem~\ref{lemma_finitelysupported}). 

In Section~\ref{subsec:relaxed_game}, we will analyze the standard idea in mathematical programming of looking at the game obtained by relaxing the integrality requirements and we will argue that this seems not to provide useful information about the original set of {\NE} for the associated {\IPG}. Hence, another perspective it taken to tackle the problem. In Section~\ref{subsec:algorithm_formalized}, we design our algorithmic scheme for computing equilibria. It iteratively tightens an inner approximation to the original {\IPG}. This framework incorporates two crucial components: an algorithm for searching an {\NE} for normal-form games and a mathematical programming solver for computing best responses. While they can be left as a choice for the user who may have specific implementations exploring problem structure, in Section~\ref{subsec:pns}, we review the Porter-Nudelman-Shoham method ({\PNS})~\citep{Porter2008642} for searching the {\NE} of normal-form games, given its practical effectiveness and flexibility to take advantage of the overall iterative methodology. The basic algorithm obtained from our framework is modified in Section~\ref{sec:modifiedSGM}, in an attempt to improve its performance. Finally, in Section~\ref{sec:extensions}, we describe the extension of our methodology to correlated equilibria. 

Before proceeding, it is worthy to clarify  that in all our experiments, we consider players payoffs of the form
\begin{equation}
    \Pi^p(x^p,x^{-p}) = (c^p)^T x^p - \frac{1}{2} (x^p)^T Q_p^p x^p + \sum_{k \in M: i\neq p} (x^k)^T Q^p_k x^p,
    \label{payoff:quadratic}
\end{equation}
\ie, separable (recall Definition~\ref{ObjectivePlayer_separable}) quadratic payoff functions with bilateral (pairwise) interactions. The correctness of our methodology follows for more general payoff function forms. Thus, the value of this remark comes instead from concrete choices of our methodology components, as we will remark along the text.


\subsection{Game relaxations}\label{subsec:relaxed_game}

A typical procedure to solve optimization problems consists in relaxing constraints that are hard to handle and to use the information associated with the relaxed problem  to guide the search for the optimum. Thus, in this context, such ideas seem a natural direction to investigate. Call \emph{relaxed integer programming game} ({\RIPG}) the game resulting from an {\IPG} when the integrality constraints are removed. In the following examples, we  compare the  {\NE} of  an {\IPG} with the ones of the associated  {\RIPG}.

\begin{example}[{\RIPG} has more equilibria than {\IPG}] Consider an instance with two players, in which player $A$ solves $$\max_{x^A} 5 x^A_1x^B_1+23x^A_2x^B_2 \textrm{ subject to } 1\leq x^A_1+3x^A_2 \leq 2 \textrm{ and }  x^A\in \lbrace 0,1 \rbrace^2$$ and player $B$ solves $$\max_{x^B} 5 x^A_1x^B_1+23x^A_2x^B_2 \textrm{ subject to } 1 \leq x^B_1+3x^B_2 \leq 2 \textrm{ and }  x^B\in \lbrace 0,1 \rbrace^2.$$ 
	
	There is only one feasible strategy for each player in the {\IPG}. Thus, it is easy to see that it has a unique equilibrium: $(x^A,x^B)=((1,0),(1,0))$. This equilibrium also holds for  {\RIPG}. However,  {\RIPG} possesses at least one more equilibrium: $(x^A,x^B)=((0,\frac{2}{3}),(0,\frac{2}{3}))$.
	\label{ExampleMoreEquilibria}
\end{example}

\begin{example}[{\RIPG} has less equilibria than {\IPG}]
	Consider the game where player A solves
	$$  \max_{x^A} 12x_1^A x_1^B+5x_2^Ax_2^B \textrm{ subject to }  2x_1^A+2x_2^A \leq 3  \textrm{ and } x^A\in \lbrace 0,1 \rbrace^2,$$
	and player B solves 
	$$  \max_{x^B}  12x_1^Ax_1^B+5x_2^Ax_2^B+100x_1^B \textrm{ subject to }2x_1^B+x_2^B \leq 1  \textrm{ and } x^B \in \lbrace 0,1 \rbrace^2.$$
	There are at least 2 equilibria: $(x^A,x^B)=((0,0),(0,0))$ and $(x^A,x^B)=((0,1),(0,1))$. However, none is an equilibrium of the associated {\RIPG}. In fact, in the {\RIPG}, it is always a dominant strategy for player $B$ to select $x^B =(\frac{1}{2},0)$, and the unique equilibrium is $(x^A,x^B)=( (1,0), (\frac{1}{2},0))$. In conclusion, the game has at least 2 equilibria while the associated relaxation has 1.
	\label{ExampleLessEquilibria}
\end{example}

These examples show that no bounds on the number of {\NE} and, thus, on the players' payoffs in an {\NE} can be extracted from the relaxation of an {\IPG}.

Moreover, there are no general methods to compute mixed equilibria of {\RIPG}s, implying that we would be restricted to pure equilibria (which may fail to exist). The exception is the case where payoffs are separable with linear functions, \ie, of the form~\eqref{ObjectivePlayer_separable} with $f^p_j$ linear. Under this setting, any mixed strategy profile of {\RIPG} can be re-written as a pure strategy profile without changing players' payoffs (see the proof of Theorem 7 in \cite{Carvalho2019} where this is shown). In other words, such {\RIPG}s are guaranteed to have pure equilibria.

\subsection{Algorithm formalization}\label{subsec:algorithm_formalized}

Our goal is to determine an {\NE}. Thus,  from the Nash equilibrium definition, we aim to find $\sigma=(\sigma^1, \ldots, \sigma^m)$ such that
\begin{subequations}
	\begin{alignat}{4}
		& \sigma^p \in \Delta(X^p)  \qquad &  \forall p \in M \\
		& \Pi^p(  \sigma^p , \sigma^{-p}) \geq  \Pi^p( x^p ,\sigma^{-p}) \qquad & \forall p \in M,  &\qquad \forall x^p \in X^p, \label{NE_inequality}
	\end{alignat}
	\label{Constraint_Programing_NE}%
\end{subequations}
that is, we aim to determine a mixed profile of strategies such that no player has incentive to unilaterally deviate from it. The number of pure strategies in each $X^p$ is likely to be infinite  or, in  case all variables are  integer and bounded, to be exponential. Moreover, even with only two players, the left-hand-side of Inequalities~\eqref{NE_inequality} is non-linear; recall the expected payoff~\eqref{expected_payoff_sigma_discrete}. Thus, in general, tackling Problem~\eqref{Constraint_Programing_NE} directly will not be possible in practice.

We then follow the motivating idea of column generation~\citep{Gomory1961} and cutting plane~\citep{gomory1958} approaches: many pure strategy profiles will be played with zero probability (Theorem~\ref{lemma_finitelysupported}) and only a subset of the Constraints~\eqref{NE_inequality} will be binding  under an equilibrium. Indeed, we will decompose an {\IPG} through its \emph{sampled games}, \ie, the {\IPG} restricted to finite subsets of $X$.

Algorithm~\ref{Alg:SGM} presents our general methodology. In \ref{initial_step_unique_equilibrium_SGM}, we obtain our first sampled game represented by the subset of pure strategy profiles $\mathbb{S}$; computationally, we use its polymatrix normal-form representation, \ie, since players' interactions are bilateral~\eqref{payoff:quadratic}, we just need to save the payoffs for all pairs of pure strategies. Then, in \ref{step:search_eq}, we search for an {\NE} $\sigma_k$ of the obtained sampled game; note that any algorithmic approach for normal-form games can be applied in this step. In \ref{verify_NE_step_SGM}, we verify if there is a player with incentive to deviate. Here, with exception to the last iteration of the algorithm, we can both determine a player best reaction~\eqref{GeneralProblem}, or use some heuristic that finds a pure strategy that does not decrease by more than $\varepsilon$ the player payoff in comparison with the sampled game equilibrium. If no player has incentive to deviate, the algorithm returns an $\varepsilon$-equilibrium. Otherwise, in \ref{verify_NE_new_sampled}, the sampled game approximation is refined by adding  the new pure strategy found in \ref{verify_NE_step_SGM}. We note that when $\varepsilon=0$, the algorithm outputs a {\NE}.  We call Algorithm~\ref{Alg:SGM} \emph{sampled generation method} ({\SGM}).

\begin{algorithm}[htbp!] \footnotesize
	\textbf{Input: }An {\IPG} instance and $\varepsilon \geq 0$. \\
	\textbf{Output: } $\varepsilon$-equilibrium, last sampled game and number of iterations.\\
	\DontPrintSemicolon
	
	\nlset{Step 1}\textbf{Initialization:}\; 
	$\mathbb{S}=\prod_{p=1}^m \mathbb{S}^p \leftarrow$ {\color{red}$Initialization(IPG)$} \hspace{0.1cm}\textcolor{gray}{// Generation of sampled game. Details in Table~\ref{Table:specialized_algorithms}}\;  \label{initial_step_unique_equilibrium_SGM}
	$k \leftarrow 0$ \;
	set $\mathbb{S}_{dev_k}$ to be $\prod_{p=1}^m \emptyset$ \hspace{0.1cm}\textcolor{gray}{// Record players deviation sets.} \;
	
	\nlset{Step 2}\textbf{Solve sampled game $k$:} \;
	$\sigma_k \leftarrow$  {\color{red}SearchNE $(\mathbb{S})$} \hspace{1.9cm}\textcolor{gray}{// Computation of {\NE}. Details in Section~\ref{subsec:pns}} \;
	$list \leftarrow$   {\color{red} $PlayerOrder(\mathbb{S}_{dev_0}, \ldots,\mathbb{S}_{dev_k})$} \hspace{0.3cm}\textcolor{gray}{// A list ordering players. Details in Table~\ref{Table:specialized_algorithms}}  \; 
	\label{step:search_eq}

	\nlset{Step 3}\textbf{Termination:}\;
	\For{$p \in list$}{
		$x(k+1) \leftarrow {\color{red} DeviationReaction(p,\sigma_k^{-p},\Pi^p (\sigma_k),\varepsilon,IPG)}$ \hspace{0.3cm}\textcolor{gray}{// Incentive to deviate. Details in Table~\ref{Table:specialized_algorithms}}\;
		\If{ $\Pi^p(\sigma_k) +\varepsilon< \Pi^p(x(k+1), \sigma^{-p}_k)$}{
			go to Step 4}
	}
	\Return $\sigma_k$, $\mathbb{S}$, $k$\; \label{verify_NE_step_SGM}
	
	\nlset{Step 4}\textbf{Generation of next sampled game:} \;
	$k \leftarrow k+1$ \;
	$\mathbb{S}_{dev_{k}}^p \leftarrow  \lbrace x(k) \rbrace$   \hspace{0.1cm}\textcolor{gray}{// Record the deviation.} \;
	$\mathbb{S}^p \leftarrow \mathbb{S}^p \cup \lbrace x(k) \rbrace$ \;
	go to Step 2\;
	\label{verify_NE_new_sampled}
	
	\caption{Sampled generation method ({\SGM}).} \label{Alg:SGM}
\end{algorithm}

Figure~\ref{Fig:ISM_iterations} illustrates in normal-form (bimatrix-form) the sampled games progressively obtained through {\SGM}. Intuitively, we expect {\SGM} to enumerate the most ``relevant'' strategies and/or ``saturate'' the space $X$ after a sufficient number of iterations and thus, approximate an equilibrium of the original {\IPG}.  Hopefully, we will  not need to enumerate all feasible strategies in order to compute an equilibrium. 
\begin{figure}
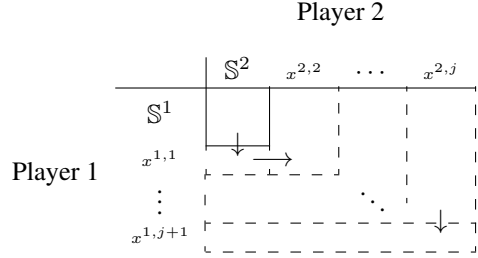
\center 
	\begin{tabular}{cccccccc}
		& & \multicolumn{5}{c}{Player 2}    \\
		&  & \multicolumn{1}{|c}{}& & &\\
		&                             & \multicolumn{2}{|c|}{$\mathbb{S}^2$} & \multicolumn{1}{c:}{\tiny $x^{2,2}$}  & $\cdots$ & \multicolumn{1}{:c:}{\tiny $x^{2,j}$}  \\ \cline{2-7}
		\multirow{6}{*}{Player 1}  & \multirow{2}{*}{$\mathbb{S}^1$}   & \multicolumn{2}{|c|}{}  &  \multicolumn{1}{c:}{}   &   &\multicolumn{1}{:c:}{}  \\ 
		& &   \multicolumn{2}{:c:}{\multirow{2}{*}{$\downarrow$}} &  \multicolumn{1}{c:}{}    &  & \multicolumn{1}{:c:}{} & \\  \cline{3-4}
		& \multicolumn{1}{c:}{\tiny $x^{1,1}$}& \multicolumn{3}{c}{$\longrightarrow$}         & &\multicolumn{1}{:c:}{}  \\ \cdashline{3-5} 
		&\multicolumn{1}{c:}{$\vdots$} &  &  &  &  $\ddots$& \multicolumn{1}{c:}{\multirow{2}{*}{$\downarrow$}}  \\  \cdashline{3-7}
		& \multicolumn{1}{c}{\tiny $x^{1,j+1}$} &  &  &  & \\ \cdashline{3-7}
		&  & &   &  &
	\end{tabular} 
	\caption{{\SGM} for $m=2$. The notation $x^{p,k}$ represents the player $p$'s strategy added at iteration $k$. A vertical (horizontal) arrow represents player 1 (player 2) incentive to unilaterally deviate from the sampled game computed equilibrium to a new strategy.}
	\label{Fig:ISM_iterations}
\end{figure}


Before providing the {\SGM}'s proof of correctness, in an attempt to clarify the method and highlight its particularities when continuous variables exist, we present the following example.

\begin{example}[Computing an equilibrium with {\SGM}]
	Consider an {\IPG} with two players. Player $i$ wishes to maximize the payoff function $\max_{x^i\geq 0} -(x^i)^2+x^ix^{-i}$. The best reaction is given by $x^i(x^{-i}) = \frac{1}{2} x^{-i}$, for $i=1,2$. The only equilibrium is $(x^1,x^2)=(0,0)$. Let us  initialize {\SGM} with the sampled game $\mathbb{S}^i = \lbrace 10 \rbrace$ for $i=1,2$, and always start by checking first if player 1 has incentive to deviate. Then, in each iteration $k$, the sampled game has the pure equilibrium $(x^{1,k},x^{2,k-1})=(\frac{5}{2^{k-1}},\frac{10}{2^{k-1}})$ if $k$ is odd and $(x^{1,k},x^{2,k-1})=(\frac{10}{2^{k-1}},\frac{5}{2^{k-1}})$ if $k$ is even. See Table~\ref{ex:smg_app} and Figure~\ref{best_reactions_example1} for an illustration of {\SGM} iterative process evolution.
	
	\begin{table}
	    \centering
	\begin{tabular}{cc|rr}
	     && \multicolumn{2}{c}{Player 2} \\ 
	     &&  10      & $x^{2,2}=\frac{5}{2}$  \\\hline
\multirow{3}{*}{Player 1}&	  10 & (0,0)   &  (-75.0, 18.75) \\
	  & $x^{1,1}=$5 & (25,-50) & (-12.5,6.25) \\
	  & $x^{1,3}=\frac{5}{4}$ & (10.9375,-87.5) & (1.5625,-3.125) \\ \hline
	\end{tabular}
	\caption{Sampled game after 3 iterations of {\SGM}.}
		\label{ex:smg_app}
		\end{table}
	
	Thus, {\SGM} converges to the equilibrium $(0,0)$. If in the input of {\SGM}, $\varepsilon= 10^{-6}$ then, after 14 iterations, {\SGM} would return an $\varepsilon$-equilibrium of the game. Remark that in this case $\varepsilon$ cannot be zero.
	\begin{figure}[th]
		\centering
		\includegraphics[scale=0.4]{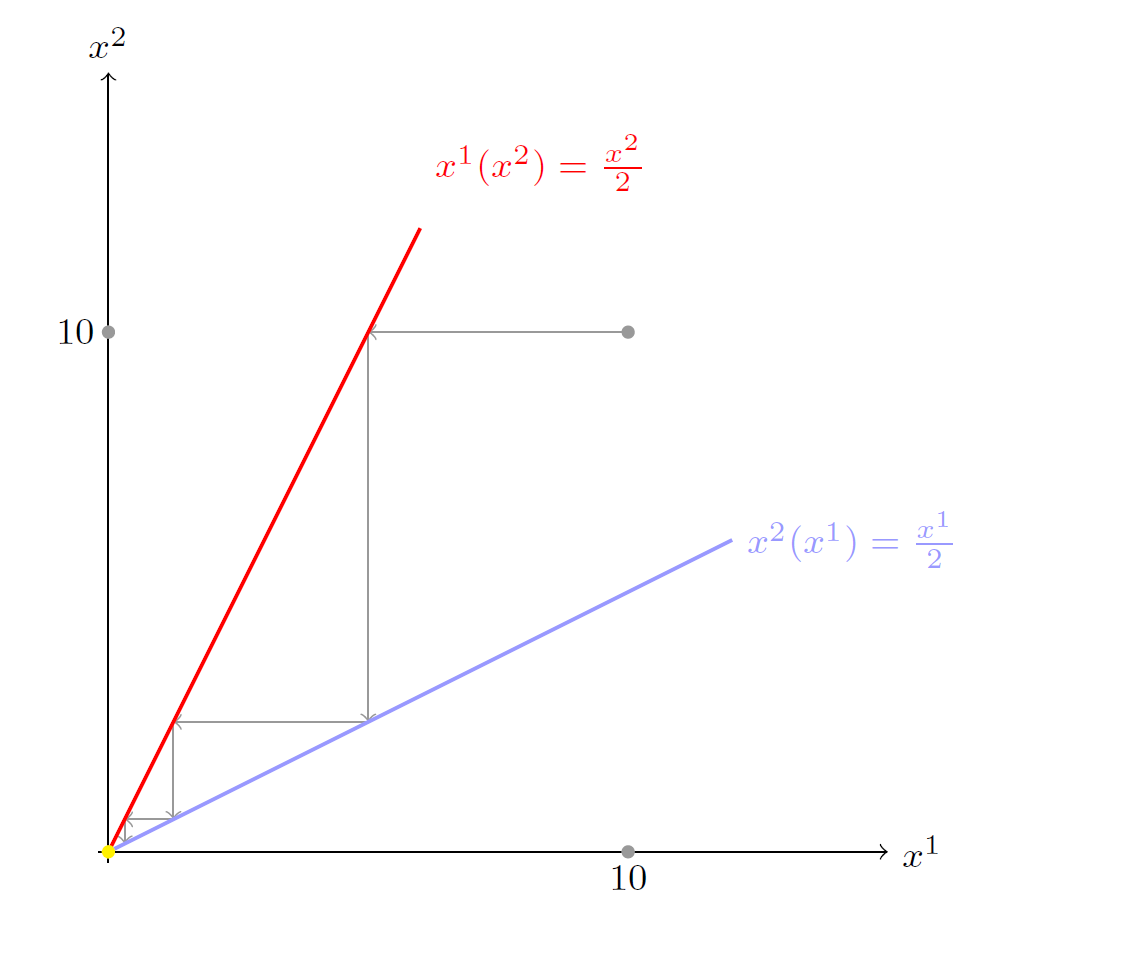}
		\caption{Players' best reaction functions. The gray line represents the players deviations in each iteration of {\SGM}.}
		\label{best_reactions_example1}
	\end{figure}
	\label{examp:algorithm_diverge_unbounded}
\end{example}

Our goal is to guarantee that {\SGM} computes an $\varepsilon$-equilibrium in finite time. To this end, some conditions on the {\IPG}s are necessary.
If a player $p$'s set of feasible strategies is unbounded, the game may fail to have equilibria, and even if it has equilibria, {\SGM} may not converge.  Furthermore, as pointed out by \cite{Stein2008} for a specific separable game, it seems that there must be some bound on the speed variation (how fast it can change) of the payoffs in order to guarantee that an algorithm computes an equilibrium in finite time.  The Lipschitz condition ensures this bound. 
\begin{theorem}
	If $X$ is nonempty and bounded, then in a finite number of steps, {\SGM} computes
	\begin{enumerate}
		\item an {\NE}, if all players' decision variables are integer;
		\item an $\varepsilon$-equilibrium with $\varepsilon>0$, if each player $p$'s payoff function is Lipschitz continuous in $X^p$.
	\end{enumerate}  
	\label{theorem_converge_equilibrium_sgm}
\end{theorem}
\begin{proof}
	
	{\SGM} stops once an equilibrium of the sampled game coincides with an equilibrium (case 1) or an $\varepsilon$-equilibrium (case 2) of the {\IPG}. 
	Suppose that the method does not stop. This means that in every iteration at least a new strategy is added to the current $\mathbb{S}$. 
	\begin{description}
		\item[Case 1:] Given that $X$ is bounded and players' variables are integer, each player has a finite number of strategies. Thus, after a finite number of iterations, the sampled game will coincide with {\IPG}, \ie, $\mathbb{S}=X$. This means that an {\NE} of the sampled game is an {\NE} of the {\IPG}.
		\item[Case 2:] Each player $p$ payoff function is Lipschitz continuous in $X^p$, which means that there is a positive real number $L^p$ such that
		$$|\Pi^p(x^p,\sigma^{-p})-\Pi^p(\hat{x}^p,\sigma^{-p})| \leq L^p \parallel x^p - \hat{x}^{p} \parallel \quad \forall x^p, \hat{x}^p \in X^p,$$
		where $ \parallel  \cdot  \parallel $ is the Euclidean norm. 
		
		Consider an arbitrary iteration of {\SGM} with $\sigma$ as an {\NE} of the current sampled game $\mathbb{S}$. If $\sigma$ is not an $\varepsilon$-equilibrium of the original game, then there is a player $p$ with incentive to deviate to  $x^p \in X^p \setminus \mathbb{S}^p$, \ie
				\begin{equation} 
					\Pi^p(x^p,\sigma^{-p})> \Pi^p(\sigma)+\varepsilon. 
				\label{eq:proof1}
				\end{equation}
				Hence, 
				$$
					\parallel x^p - \hat{x}^p \parallel >\frac{\varepsilon}{L^p}, \quad \forall \hat{x}^p \in \mathbb{S}^p.
				$$
				Otherwise, there is $\hat{x}^p \in \mathbb{S}^p$ such that $\parallel x^p - \hat{x}^p \parallel \leq \frac{\varepsilon}{L^p}$ which contradicts Inequality~\ref{eq:proof1}:
					\begin{subequations}
						\begin{alignat*}{4}
							& \Pi^p(x^p, \sigma^{-p}) - \Pi^p(\sigma)  &=& \Pi^p(x^p, \sigma^{-p})-\Pi^p(\hat{x}^p, \sigma^{-p})+\Pi^p(\hat{x}^p, \sigma^{-p}) - \Pi^p(\sigma) \\[0.4ex]
							&                                                                   & \leq &\Pi^p(x^p, \sigma^{-p})-\Pi^p(\hat{x}^p, \sigma^{-p}) \\
							&                                                                   & \leq & |\Pi^p(x^p, \sigma^{-p})-\Pi^p(\hat{x}^p, \sigma^{-p})  |\\
							&                                                                   &  \leq & L^p \parallel x^p - \hat{x}^{p} \parallel \leq L^p \frac{\varepsilon}{L^p} = \varepsilon.
						\end{alignat*}
					\end{subequations}
					The first step follows from the fact that $\sigma$ is an {\NE} of the sampled game and thus $\Pi^p(\hat{x}^p, \sigma^{-p}) \leq \Pi^p(\sigma)$. The next inequality holds because we are just applying the absolute value. The third step follows from the fact that player $p$'s payoff is Lipschitz continuous in $X^p$.
					
		Consequently, in each iteration of {\SGM}, each newly added strategy $x^p$ to $\mathbb{S}^p$ is more than $\frac{\varepsilon}{L^p}$ away from any other strategy in $\mathbb{S}^p$. Thus, after a sufficiently large number of iterations (if {\SGM} has not stopped), it holds
		$$\parallel x^p - \hat{x}^p \parallel \leq \frac{\varepsilon}{L^p}, \quad \forall p \in M, \forall x^p \in X^p.$$
		Therefore, given an {\NE} of the current sampled game, no player can improve her payoff more than $\varepsilon$ through unilateral deviations.
		 In this way, $\sigma$ is an $\varepsilon$-equilibrium of the {\IPG}. \xqed
	\end{description}
	
\end{proof}

A payoff function which is linear in that player's variables is Lipschitz continuous; a quadratic payoff function when restricted to a bounded set satisfies the Lipschitz condition. In Section~\ref{sec:competitive_lotsizing_game}, we will describe the knapsack game and the kidney exchange game which have linear payoffs, and the lot-sizing game which has quadratic payoffs. Therefore, Lipschitz continuity seems not to be too restrictive in practice.

\subsubsection{Computation of {\NE} for normal-form games}\label{subsec:pns}

 A relevant fact about computing equilibria for a sampled game with the set of strategies $\mathbb{S} \subseteq X$ is that  $\mathbb{S}$ is finite and, consequently, enables the use of general algorithms to compute equilibria of normal-form games. Given the good results achieved by {\PNS}~\citep{Porter2008642} for the computation of a {\NE} in normal-form games, this is the method that our framework will apply to solve the sampled games (additional advantages for adopting {\PNS} will be given in the end of this section). {\PNS} solves the constrained program~\eqref{Constraint_Programing_NE} associated with a sampled game (\ie, $X=\mathbb{S}$)  through the resolution of simpler subproblems. Note that in constraints~\eqref{NE_inequality} the expected payoffs~\eqref{expected_payoff_sigma_discrete}  are highly non-linear due to the multiplication of the probability variables. To this end,  {\PNS} bases its search for an equilibrium $\sigma$ by guessing its support and using the fact that in an equilibrium $\sigma \in \Delta$, each player must be indifferent among her strategies in the support at which her payoff is maximized (Constraints~\eqref{eq:indif} and~\eqref{eq:NEmax}). Thus, an equilibrium  $\sigma$ of a sampled game $\mathbb{S}$ satisfies
\begin{subequations}
	\begin{alignat}{4}
		v^p=&  \Pi^p(\hat{x}^p, \sigma^{-p}) \quad &\forall p \in M, \quad & \forall \hat{x}^p \in \supp(\sigma^p) \label{eq:indif}\\
		v^p  \geq & \Pi^p(x^p,\sigma^{-p}) \quad & \forall p \in M, \quad & \forall x^p \in \mathbb{S}^p  \label{eq:NEmax}\\
		\sum_{x^p \in \supp(\sigma^p)}\sigma^p(x^p)  = &1 \quad & \forall p \in M \quad &\\
		\sigma^p(x^p)& \geq 0 & \forall p \in M, \quad & \forall x^p \in \supp(\sigma^p)   \\
		\sigma^p(x^p) & = 0 & \forall p \in M, \quad & \forall x^p \in \mathbb{S}^p
		\setminus\supp(\sigma^p),
	\end{alignat}
	\label{FeasibilityProgram}%
\end{subequations}
with $\supp(\sigma^p) \subseteq \mathbb{S}^p$ and $v^p$ an auxiliary variable to represent player $p$ maximum payoff.
Problem \eqref{FeasibilityProgram} is called \emph{Feasibility Problem}. When the payoff functions have the form~\eqref{payoff:quadratic}, the constraints in Problem \eqref{FeasibilityProgram} become linear, and thus, it becomes solvable in polynomial time.

The computation of an {\NE} to the sampled game $\mathbb{S}$ reduces to \emph{(i)} finding an {\NE} support and \emph{(ii)} solving the associated Feasibility Problem. Therefore, support sets in $\mathbb{S}$ are enumerated and the corresponding Feasibility Problems are solved, until an {\NE} is found (\ie, a Feasibility Problem is proven to be feasible). {\PNS} implements this enumeration with an additional step that decreases the number of Feasibility Problems to be solved, in other words, it reduces the number of candidates to be the support of an equilibrium.
A strategy $x^p \in X^p$ is \emph{conditionally dominated} given a subset of strategies $R^{-p} \subseteq X^{-p}$ for the remaining players, if the following condition holds
\begin{equation}
	\exists \hat{x}^p \in X^p \ \ \forall x^{-p} \in R^{-p}: \quad \Pi^p(x^p,x^{-p}) < \Pi^p(\hat{x}^p,x^{-p}).
	\label{def:conditionally_dominated}
\end{equation}
{\PNS} prunes the support enumeration search by making use of conditionally dominated strategies, since such strategies will never be selected with positive probability in an equilibrium. In addition, we consider in the support enumeration the property given by Theorem~\ref{lemma_finitelysupported}: each player $p$ has a support size of at most $k_p+1$; recall that to determine $k_p+1$, one just needs write player p's payoff as in the form~\eqref{ObjectivePlayer_separable}. 

We conclude {\SGM} description by highlighting an additional advantage of {\PNS}, besides being in practice the fastest algorithm. The authors' implementation of {\PNS}~\citep{Porter2008642} searches the equilibria by following a specific order for the enumeration of the supports. In specific, for two players' games, $\vert M \vert = 2$, the algorithm starts by enumerating supports, first, by increasing order of their total size and, second, by increasing order of their balance (absolute difference in the players' support size). The idea is that in the case of two players, each equilibrium is likely to have supports with the same (small) size. When $\vert M \vert>2$, {\PNS} exchanges the importance of these two criteria. We expect {\SGM} to start converging to an equilibrium as it progresses. Therefore, it may be advantageous to use the past computed equilibria to guide the support enumeration. Including rules for support enumeration in {\PNS} is straightforward; these rules can be problem specific. On the other hand, doing so for other state-of-the-art algorithms is not as easy. For instance, the well-known \cite{Lemke1964} algorithm implies to start the search for equilibria in an artificial equilibrium or in an equilibrium of the game (allowing to compute a new one). Thus, since at iteration $k$ of {\SGM}, none of the equilibria computed for the sampled games in iterations $1$ to $k-1$ is an {\NE} of the current sampled game, there is no direct way of using past information to start or guide the Lemke-Howson algorithm. Moreover, this algorithm's search is performed by enumerating vertices of polytopes built according to the game strategies. Therefore, since in each iteration of {\SGM} a new strategy is added to the sampled game, these polytopes may change deeply.

\subsubsection{Modified {\SGM}}\label{sec:modifiedSGM}

Based on the framework described, we can slightly change the scheme of {\SGM} presented in Algorithm~\ref{Alg:SGM}, in an attempt to speed up its running time. Its new version will be a depth-first search: while in {\SGM} the size of the sampled game strictly increases from one iteration to the next one, in its depth-first search version it will be possible to backtrack to previous sampled games, with the aim of decreasing the size of the sampled game. In each iteration of the improved {\SGM}, we search for an equilibrium which has in the support the last strategy added to the sampled game; in case such equilibrium does not exist, the method backtracks, and computes a new equilibrium to the previous sampled game. While in each iteration of {\SGM} all supports can be considered, in the \emph{modified} {\SGM} ({\modSGM}) we limit the search to the ones with the new added strategy.  Therefore, this modified {\SGM} attempts to keep the size of the sampled game small and decreases the number of supports enumerated.

Next, we concentrate in proving under which conditions the {\modSGM} computes an $\varepsilon$-equilibrium in finite time and provide its detailed description.

\begin{theorem}
	Let $\mathbb{S} = \mathbb{S}^1 \times \mathbb{S}^2 \times \ldots \times \mathbb{S}^m$ represent a sampled game associated with some {\IPG}. If the finite game that results from $\mathbb{S}$ has a unique equilibrium $\sigma$, then one of the following implications holds:
	\begin{enumerate}
		\item $\sigma$ is an equilibrium of the {\IPG};
		\item given any player $p$ with incentive to deviate from $\sigma^p$ to $x^p \in X^p$, the finite game game associated with $\mathbb{S}'=\mathbb{S}^1 \times \cdots \times \mathbb{S}^{p-1} \times\left( \mathbb{S}^p \cup \lbrace x^p \rbrace \right) \times \mathbb{S}^{p+1} \times \cdots \times \mathbb{S}^m$  has $x^p$ in the support of all its equilibria.
	\end{enumerate}
	\label{the:uniqueNE}
\end{theorem}
\begin{proof}
	Suppose $\sigma$ is not an equilibrium of the {\IPG}. Then, by the definition of equilibrium, there is a player, say player $p$, with incentive to unilaterally deviate to some $x^p \in X^p\setminus\mathbb{S}^p$. By contradiction, assume that there is an equilibrium $\tau$ in $\mathbb{S}'$ such that $x^p$ is played with zero probability (it is not in the support of $\tau$). First, $\tau$ is different from $\sigma$ because now $\mathbb{S}'$ contains $x^p$. Second, $\tau$ is an equilibrium for the game restricted to $\mathbb{S}$, contradicting the fact that $\sigma$ was its unique equilibrium.\xqed
\end{proof}

In this way, if in an iteration of {\SGM} the sampled game has an unique {\NE}, in the subsequent iteration, we can prune the support enumeration search of {\PNS} by forcing the new strategy added to be in the support of the {\NE} to be computed. Note that it might occur that in the consecutive sampled games there is more than one {\NE} and thus, an equilibrium with the new added strategy in the support may fail to exist (Theorem~\ref{the:uniqueNE} does not apply). Therefore, backtracking is introduced so that a previously processed sampled game can be revisited and its support enumeration continued in order to find a new {\NE} and to follow a promising direction in the search. In  Algorithm~\ref{Algorithm:NashCuttingPlane_Improved},  {\modSGM} is described. The subroutines called by it are described in Table~\ref{Table:specialized_algorithms} and can be defined independently. We will propose an implementation of them in Section~\ref{subsec:scpecialized_functions}. 

\begin{algorithm}[htbp!] \footnotesize
	\textbf{Input: }An {\IPG} instance and $\varepsilon\geq0$. \\
	\textbf{Output: } $\varepsilon$-equilibrium, last sampled game and number of the last sampled game.\\
	\DontPrintSemicolon
	
	\nlset{Step 1}\textbf{Initialization:}\; 
	$\mathbb{S}=\prod_{p=1}^m \mathbb{S}^p \leftarrow$ {\color{red}$Initialization(IPG)$}\;  \label{initial_step_unique_equilibrium}
	$k \leftarrow 0$ \;
	set $\mathbb{S}_{dev_k}$, $\mathbb{S}_{dev_{k+1}}$ and $\mathbb{S}_{dev_{k+2}}$ to be $\prod_{p=1}^m \emptyset$ \;
	$\sigma_k \leftarrow (1,\ldots,1)$ is Nash equilibrium of the current sampled game $\mathbb{S}$ \;
	$list \leftarrow$   {\color{red}$PlayerOrder(\mathbb{S}_{dev_0}, \ldots,\mathbb{S}_{dev_k})$}  \;  \label{initial_step}
	
	\nlset{Step 2}\textbf{Termination:}\;
	\While{$list$ non empty}{
		$p \leftarrow list.pop()$ \;
		$x(k+1) \leftarrow \color{red} DeviationReaction(p,\sigma_k^{-p},\Pi^p (\sigma_k),\varepsilon,IPG)$ \;
		\If{ $\Pi^p(\sigma_k) +\varepsilon< \Pi^p(x(k+1), \sigma^{-p}_k)$}{
			go to Step 3}
	}
	\Return $\sigma_k$, $\mathbb{S}$, $k$\; \label{verify_NE_step}
	
	\nlset{Step 3}\textbf{Generation of next sampled game:} \;
	$k \leftarrow k+1$ \;
	$\mathbb{S}_{dev_{k}}^p \leftarrow \mathbb{S}_{dev_{k}}^p \cup \lbrace x(k) \rbrace$ \;
	$\mathbb{S}^p \leftarrow \mathbb{S}^p \cup \lbrace x(k) \rbrace$ \;
	$\mathbb{S}_{dev_{k+2}} \leftarrow \prod_{p=1}^m \emptyset$\;  \label{algorithm_forward}

	\nlset{Step 4}\textbf{Solve sampled game $k$:} \;
	$Sizes_{ord} \leftarrow SortSizes(\sigma_0,\ldots,\sigma_{k-1})$\;
	$Strategies_{ord} \leftarrow SortStrategies(\mathbb{S}, \sigma_0,\ldots,\sigma_{k-1})$ \;
	$\sigma_k \leftarrow$ {\color{red}{\PNS}$_{adaptation}(\mathbb{S},x(k),\mathbb{S}_{dev_{k+1}},Sizes_{ord},Strategies_{ord})$} \;
	\If{ {\PNS}$_{adaptation}(\mathbb{S},x(k),\mathbb{S}_{dev_{k+1}},Sizes_{ord},Strategies_{ord} )$ fails to find equilibrium}{
		$\mathbb{S} \leftarrow \mathbb{S} \backslash \mathbb{S}_{dev_{k+1}}$ \;
		remove from memory $\sigma_{k-1}$ and $\mathbb{S}_{dev_{k+2}}$ \;
		$k \leftarrow k-1$ \;
		go to Step 4 (backtrack)\; 
	}
	\Else{
		$list \leftarrow$ {\color{red}$PlayerOrder(\mathbb{S}_{dev_0}, \ldots,\mathbb{S}_{dev_k})$}\;
		go to Step 2
	}  \label{last_step4}
	\caption{Modified {\SGM} ({\modSGM}).} \label{Algorithm:NashCuttingPlane_Improved}
\end{algorithm}

{\renewcommand{\arraystretch}{2}
	\begin{table} \tiny
		\centering
		\begin{tabularx}{\textwidth}{|lX|}
			\hline
			ALGORITHMS &  DESCRIPTION\\[1ex]
			\hline 
			$Initialization(IPG)$ & Returns sampled game of the {\IPG} with one feasible strategy for each player. \\ [1ex]
			\hline 
			$PlayerOrder(\mathbb{S}_{dev_0}, \ldots,\mathbb{S}_{dev_k})$ & Returns a list of the players order that takes into account the algorithm history.  \\ [1ex]
			\hline 
			$DeviationReaction(p,\sigma_k^{-p},\Pi^p (\sigma_k),\varepsilon,IPG)$ & If there is $x^p \in X^p$ such that $\Pi^p(x^p,\sigma_k^{-p})>\Pi^p(\sigma_k)+\varepsilon$, returns $x^p$; otherwise, returns any player $p$'s feasible strategy.  \\ [1ex]
			\hline 
			$SortSizes(\sigma_0,\ldots,\sigma_{k-1})$ & Returns an order for the support sizes enumeration that takes into account the algorithm history. \\ [1ex]
			\hline 
			$SortStrategies(\mathbb{S}, \sigma_0,\ldots,\sigma_{k-1})$  &  Returns an order for the players' strategies in $\mathbb{S}$ that takes into account the algorithm history.  \\ [1ex]
			\hline 
			{\PNS}$_{adaptation}(\mathbb{S},x(k),\mathbb{S}_{dev_{k+1}},Sizes_{ord},Strategies_{ord})$ & Applies {\PNS} in order to return a Nash equilibrium $\sigma$ of the sampled game $\mathbb{S}$ of the {\IPG} such that $x(k) \in supp(\sigma)$ and $\mathbb{S}_{dev_{k+1}} \cap supp(\sigma) = \emptyset$; It makes the support enumeration according with $Sizes_{ord}$ and $Strategies_{ord}$. \\ [1ex]
			\hline
		\end{tabularx} 
		\caption{Specialized subroutines. The algorithm history keeps record of the order in which strategies where added to the sampled game.}
		\label{Table:specialized_algorithms}
	\end{table}
}

Figure~\ref{fig:modified_sgm_sampledgames} illustrates {\modSGM}.
\begin{figure}[t]
	\centering
	\begin{tikzpicture}[->,>=stealth',shorten >=1pt,auto,scale=0.5,node distance=4.5cm,
		thick,main node/.style={scale=0.8,rectangle,draw,white,font=\sffamily\Large\bfseries}]
		
		\node[main node] (1) {\colorbox{blue!20}{$\quad$ \hspace{1.3cm}} };
		\node[main node] (2) [below of=1] {
			\begin{tabular}{c}
				\colorbox{blue!20}{$\quad$ \hspace{1.3cm}} \\ 
				\colorbox{pink}{\color{black}  $\mathbb{S}_{dev_1}$ \tiny \colorbox{pink!20}{\color{black} $x(1)$}}
			\end{tabular}
		};
		\node[main node] (3) [below of=2] {
			\begin{tabular}{c}
				\colorbox{blue!20}{$\quad$  \hspace{1.3cm}} \\
				\colorbox{pink}{ \color{black}  $\mathbb{S}_{dev_1}$ \hspace{0.7cm}} \\ 
				\colorbox{green!40}{ \color{black}  $\mathbb{S}_{dev_2}$ \tiny \colorbox{green!20}{\color{black} $x(2)$}} \\ 
			\end{tabular}  
		};
		\node[main node]  (4) [right of=2] {
			\begin{tabular}{c}
				\colorbox{blue!20}{$\quad$  \hspace{1.3cm}} \\
				\colorbox{pink}{ \color{black}  $\mathbb{S}_{dev_1}$ \hspace{0.7cm}} \\ 
				\colorbox{green!40}{ \color{black}  $\mathbb{S}_{dev_2}$ \hspace{0.5cm} } \\ 
			\end{tabular}  
		};
		\node[main node] (5) [below of=4] {
			\begin{tabular}{c}
				\colorbox{blue!20}{$\quad$  \hspace{1.3cm}} \\
				\colorbox{pink}{ \color{black}  $\mathbb{S}_{dev_1}$ \hspace{0.7cm}} \\ 
				\colorbox{green!40}{ \color{black}  $\mathbb{S}_{dev_2}$  \hspace{0.7cm} \tiny \colorbox{green!20}{\color{black} $x(2)$}} \\ 
			\end{tabular}       
		};
		\node[main node] (6) [below of=5] {
			\begin{tabular}{c}
				\colorbox{blue!20}{$\quad$  \hspace{1.3cm}} \\
				\colorbox{pink}{  \color{black} $\mathbb{S}_{dev_1}$ \hspace{0.7cm}} \\ 
				\colorbox{green!40}{ \color{black}  $\mathbb{S}_{dev_2}$  \hspace{1.4cm}} \\ 
				\colorbox{yellow!40}{\color{black}   $\mathbb{S}_{dev_3}$  \hspace{0.7cm} \tiny \colorbox{yellow!20}{\color{black} $x(3)$}}\\
			\end{tabular}       
		};
		\node[main node] (7) [right of=5] {
			\begin{tabular}{c}
				\colorbox{blue!20}{$\quad$  \hspace{1.3cm}} \\
				\colorbox{pink}{\color{black}   $\mathbb{S}_{dev_1}$ \hspace{0.7cm}} \\ 
				\colorbox{green!40}{ \color{black}  $\mathbb{S}_{dev_2}$  \hspace{1.5cm}} \\ 
				\colorbox{yellow!40}{ \color{black}  $\mathbb{S}_{dev_3}$ \hspace{1.5cm}}\\
			\end{tabular}         
		};
		\node[main node] (8) [right of=4] {
			\begin{tabular}{c}
				\colorbox{blue!20}{$\quad$  \hspace{1.3cm}} \\
				\colorbox{pink}{\color{black}   $\mathbb{S}_{dev_1}$ \hspace{0.7cm}} \\ 
				\colorbox{green!40}{  \color{black}  $\mathbb{S}_{dev_2}$  \hspace{1.5cm}} \\ 
			\end{tabular}      
		};
		\node[main node] (9) [right of=8]{\color{black} \footnotesize Sampled game 1};
		\node[main node] (10) [above of=9]{\color{black}\footnotesize Sampled game 0};
		\node[main node] (11) [below of=9]{\color{black}\footnotesize Sampled game 2};
		\node[main node] (12) [below of=11]{\color{black}\footnotesize Sampled game 3};

		\path[every node/.style={font=\sffamily\small}]
		(1) edge node {$x(1)$} (2)
		(2) edge node {$x(2)$} (3)
		(3) edge node [right,near start] {\tiny $backtracking$} (4)
		(4) edge node  {$x(2)$} (5)
		(5) edge node {$x(3)$} (6)
		(6) edge node [right,near start] {\tiny $backtracking$} (7)
		(7) edge node {\tiny $backtracking$} (8);
	\end{tikzpicture}
	\caption{Illustration of the sampled games generated by modified {\SGM} during its execution.}
	\label{fig:modified_sgm_sampledgames}
\end{figure}
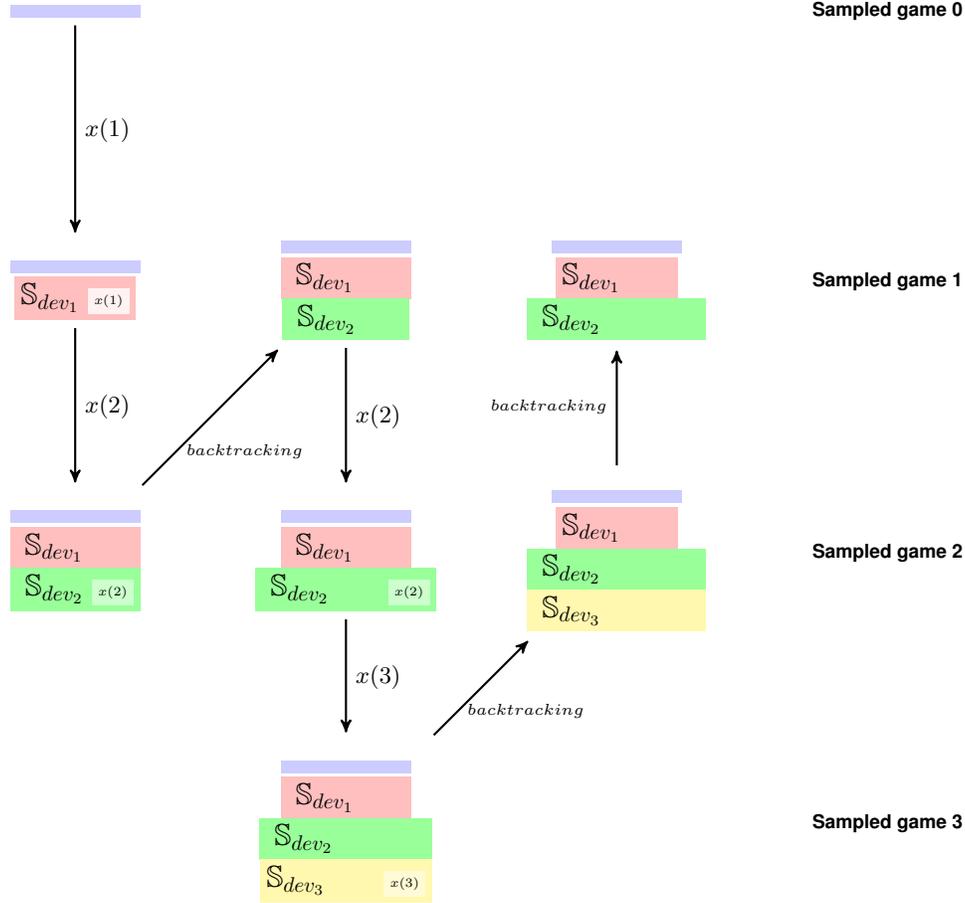
Fundamentally, whenever  {\modSGM} moves \emph{forward} (Step 3) a new strategy $x(k+1)$ is added to the sampled game $k$ that is expected to be in the support of the equilibrium of that game (due to Theorem~\ref{the:uniqueNE}). For the sampled game $k$, if the algorithm fails to compute an equilibrium with $x(k)$ in the support and $\mathbb{S}_{dev_{k+1}}$ not in the supports (see ``if'' part of Step 4) the algorithm \emph{backtracks}: it \emph{revisits} the sampled game $k-1$ with $\mathbb{S}_{dev_{k}}$ added, so that no equilibrium is recomputed. It is crucial for the correctness of the {\modSGM} that it starts from a sampled game of the {\IPG} with an unique equilibrium. To this end, the initialization determines one feasible solution for each player.
See example \ref{example_backtracking} in the Appendix \ref{app:example_backtracking} to clarify the application of {\modSGM}.

Next,  {\modSGM}  correctness will be proven.
\begin{lemma}
	In the {\modSGM}, the sampled game $0$ is never revisited.
	\label{lemma:uniqueNE_not_revisited}
\end{lemma}
\begin{proof}
	If the sampled game $0$ is revisited, it would be because the algorithm backtracks. Suppose that at some sampled game $k>0$, the algorithm consecutively backtracks up to the sampled game $0$. Consider the first sampled game $j<k$ that is revisited in this consecutive bactracking such that the last time that it was built by the algorithm it had an unique equilibrium where $x(j)$ was in the support and its successor, sampled game $j+1$, had multiple equilibria. By Theorem~\ref{the:uniqueNE}, when the algorithm moves forward from this sampled game $j$ to $j+1$, all its equilibria have $x(j+1)$ in their support. Therefore, at this point, the {\modSGM} successfully computes an equilibrium and moves forward. The successor, sampled game $j+2$, by construction, has 
	at least 
	one equilibrium and all its equilibria must have $x(j+1)$ or $x(j+2)$ in the supports. Thus, either the algorithm \emph{(case 1)} backtracks to the sampled game $j+1$ or \emph{(case 2)} proceeds to the sampled game $j+3$.  
	In  case 1, the algorithm successfully computes an equilibrium with $x(j+1)$ in the support and without $x(j+2)$ in the support, since the backtracking proves that there is no equilibrium with $x(j+2)$ in the support and, by construction, the sampled game $j+1$ has multiple equilibria. Under case 2, the same reasoning holds: the algorithm will backtrack to the sampled game $j+2$ or move forward to the sampled game $j+3$. In this way, because of the multiple equilibria in the successors of sampled game $j$,  the algorithm will never be able to return to the sampled game $j$ and thus, to the sampled game $0$. \xqed
\end{proof}

Observe that when a sampled game $k-1$ is revisited, the algorithm only removes the strategies $\mathbb{S}_{dev_{k+1}}$ from the current sampled game $k$ (``if'' part of Step 4). This means that in comparison with the last time that the algorithm built the sampled game $k-1$, it has the additional strategies $\mathbb{S}_{dev_{k}}$.  Therefore, there was a strict increase in the size of the sampled game $k-1$.

\begin{lemma}
	There is a strict increase in the size of the sampled game $k$ when the {\modSGM} revisits it.
	\label{lemma_revisiting_increase}
\end{lemma}

\begin{corollary}
	If $X$ is nonempty and  bounded, then in a finite number of steps, {\modSGM}  computes
	\begin{enumerate}
		\item an equilibrium if all players' decision variables are integer;
		\item an $\varepsilon$-equilibrium with $\varepsilon>0$, if each player $p$'s payoff function is Lipschitz continuous in $X^p$.
	\end{enumerate}  
	\label{coro:correctness}
\end{corollary}
\begin{proof}
	The \textbf{while} of  Step 2 ensures that when the algorithm stops, it returns an equilibrium (case 1) or $\varepsilon$-equilibrium (case 2). Since by Lemma~\ref{lemma:uniqueNE_not_revisited} the algorithm does not revisit sampled game $0$, it does not run into an error\footnote{By construction sampled game $0$ as a unique equilibrium.}. Moreover, if the algorithm is moving forward to a sampled game $k$, there is a strict increase in the size from the sampled game $k-1$ to $k$. Likewise, if the algorithm is revisiting a sampled game $k$, by Lemma~ \ref{lemma_revisiting_increase}, there is also a strict increase with respect to the previous sampled game $k$. Therefore, applying the reasoning of Theorem~\ref{theorem_converge_equilibrium_sgm} proof, {\modSGM} will compute an equilibrium (case 1) or $\varepsilon$-equilibrium (case 2) in a finite number of steps.\xqed
\end{proof}

%

Algorithm {\modSGM} is initialized with a sampled game that contains one strategy for each player which ensures that its equilibrium is unique. However, note that in our proof of the algorithm correctnes,s any initialization with a sampled game with a unique equilibrium is valid. Furthermore, {\modSGM} might be easily adapted in order to be initialized with a sampled game containing more than one {\NE}. In the adaptation, backtracking to the sampled game $0$ can occur and thus, the {\PNS} support enumeration must be total, this is, all {\NE} of the sampled game $0$ must be feasible. The fundamental reasoning is similar to the one of the proof of Lemma \ref{lemma:uniqueNE_not_revisited}: if there is backtracking up to the initial sampled game $0$, it is because it must contain an {\NE} not previously computed, otherwise the successors would have successfully computed one.

\subsection{Extensions: correlated equilibria}\label{sec:extensions}

The {\SGM} framework can be easily adapted to determine other game solution concepts.
 For example, one may aim to compute a well-supported $\varepsilon$-equilibrium, \ie, a profile of strategies $\sigma \in \Delta$ where each player pure strategy in the support is an $\varepsilon$-best response. This would require to simply change \ref{step:search_eq} of {\SGM} that computes an equilibrium of the sampled game to compute a well-supported  $\varepsilon$-equilibrium.  Concretely, if one is using {\PNS}, in Problem~\eqref{ExampleMoreEquilibria}, we would replace Constraints~\eqref{eq:indif} and~\eqref{eq:NEmax} by
\begin{subequations}
	\begin{alignat}{4}
		\Pi^p(\hat{x}^p,\sigma)  \geq & \Pi^p(x^p,\sigma)-\varepsilon \quad &\forall p \in M, \forall \hat{x}^p \in \supp(\sigma^p), \forall x^p \in \mathbb{S}^p.
	\end{alignat}
	\label{WellSuppEquilibrium}%
\end{subequations}

Alternatively, in a two-player {\IPG}, one could aim to compute a $\frac{1}{2}$-equilibrium to take advantage of the existence of a linear time algorithm to compute a  $\frac{1}{2}$-equilibrium for  normal-form games, reducing considerably the computational time of {\SGM} as the most costly step will be replaced by a linear time algorithm~\citep{Daskalakis2006}.

Another important solution concept is the one of correlated equilibrium, introduced in Section~\ref{sec:problemStat}.  The main factor that distinguishes these two definitions is that in correlated equilibria are not restricted to independent distributions for each player. In fact, the set of correlated equilibria contains the set of Nash equilibria. This difference considerably decreases the difficulty of determining correlated equilibria. For instance, compare the expected payoffs~\eqref{expected_payoff_sigma_discrete} and~\eqref{expected_payoff_tau_discrete}: the first is highly non-linear in $\sigma \in \Delta$, while the second is linear in $\tau \in \Delta(X)$.
This is the motivation behind PNS which by fixing a potential support for a Nash equilibrium already eliminates a term from the non-linearity in~\eqref{expected_payoff_sigma_discrete}. On the other hand,  $ \prod \sigma^i(x^i)$ is replaced by a unique probability $\sigma(x)$ when we consider correlated equilibria. Correlated equilibria can be interpreted as a third party signaling the players on what they should do; this is a reasonable assumption in many applications where players have access to news, historical behavior, etc.  In the framework of {\SGM}, \ref{step:search_eq} and \ref{verify_NE_step_SGM} must be changed in order to compute a {\CE} for an {\IPG}.

In \ref{step:search_eq}, we must compute a {\CE} of the sampled game. Mathematically, $\tau \in \Delta(\mathbb{S})$ is a correlated equilibrium of a sampled game $\mathbb{S}$ if
\begin{subequations}
	\begin{alignat}{4} 
		\sum_{x \in \mathbb{S}: x^p=\bar{x}^p} \Pi^p(x) \tau(x)\geq & \sum_{x \in \mathbb{S}: x^p=\bar{x}^p}  \Pi^p(\hat{x}^p,x^{-p})  \tau(x)\ \ &\forall p \in M,& \forall \bar{x}^p,\hat{x}^p \in  \mathbb{S}^p  \label{eq:co_max}\\
		\sum_{x \in \mathbb{S}}\tau(x)  = &1 \\
		\tau(x)\geq& 0  & \forall x \in \mathbb{S}.
	\end{alignat}
	\label{Correlated_LP}%
\end{subequations}
Note that all constraints are linear. Hence, we can even  add a linear objective function allowing to compute the associated optimal correlated equilibrium without increasing this step time complexity. For example, one could compute the correlated equilibrium that maximizes the social welfare
$$ \sum_{p \in M} \Pi^p(\tau)=\sum_{p \in M}  \sum_{x \in \mathbb{S}}  \Pi^p(x) \tau(x). $$

In \ref{verify_NE_step_SGM}, the instructions inside the cycle for player $p$ are also modified. For each $\bar{x}^p \in \mathbb{S}^p$, we must solve
\begin{equation}
\Pi^p_*=\max_{\hat{x}^p \in X^p} \sum_{x \in \mathbb{S}: x^p=\bar{x}^p}  \Pi^p(\hat{x}^p,x^{-p}) \tau(x),
\label{Problem:Correlated}
\end{equation}
\ie, compute player $p$ best response to $\tau$ when she is ``advised' by the third party to play $\bar{x}^p$.
If $\Pi^p_* > \sum_{x \in \mathbb{S}: x^p=\bar{x}^p} \Pi^p(x) \tau(x) $, then Constraint~\eqref{eq:co_max} is not satisfied for the {\IPG} and hence, the computed  strategy by solving Problem~\eqref{Problem:Correlated} must be added to the sampled game. In fact, it is easy to see that we can reduce this verification step to the $\bar{x}^p \in \supp(\tau)$. 

Once a correlated equilibrium to {\IPG} has been obtained, we also verify whether it gives origin to a Nash equilibrium:

\begin{definition}
For $\tau \in \Delta(X)$, a $\tau$-based Nash equilibrium is a Nash equilibrium  $\sigma$ where for each player $p \in M$, $\supp(\sigma^p) \subseteq \lbrace \bar{x}^p\in \mathbb{S}^p: \sum_{x \in \mathbb{S}: x^p=\bar{x}^p} \tau(x) >0 \rbrace $ and $\Pi^p(\sigma)=\Pi^p(\tau)$\footnote{Since $\sigma$ is a Nash equilibrium, $\Pi^p(\sigma)=\Pi^p(\tau)$ is equivalent to $\Pi^p(x^p,\sigma^{-p})=\Pi^p(\tau), \forall x^p \in \supp(\sigma^p)$. We use this in the feasibility problem to compute $\tau$-based Nash equilibria when $\tau$ is a correlated equilibrium.}.
\label{def:ce_based_nash}
\end{definition}

Finally, we note that in~\cite{Steinthesis} (Theorem 3.3.6), it is shown that for separable games there is a {\CE} described by a finitely supported distribution.

\section{Computational investigation}\label{sec:Computational_investigation}

Section~\ref{subsec:games_study} presents the three (separable) simultaneous {\IPG}s, the knapsack game, the kidney exchange game and the competitive  lot-sizing game, in which {\SGM} and  {\modSGM} will be tested for the computation of {\NE} and {\CE}.  In Section~\ref{subsec:scpecialized_functions}, our implementations of  the specific components in Table~\ref{Table:specialized_algorithms} are described, which have practical influence in the algorithms'	performance. Our algorithms are validated in Section~\ref{subsec:computational_results} by computational results on instances of the three presented {\IPG}s. Our instances and implementations are publicly available\footnote{\url{https://github.com/mxmmargarida/IPG}}.

\subsection{Case studies}\label{subsec:games_study}
Next, the three games in which we test our algorithms are described: the knapsack game, the simplest purely integer programming game that one could devise, the kidney exchange game and the competitive lot-sizing game whose practical applicability is discussed.  

\subsubsection{Knapsack game. } One of the most simple and natural {\IPG}s would be one with each player's payoff function being linear in her variables. This is our main motivation to analyze the knapsack game. Under this setting, each player $p$ aims to solve
\begin{subequations}
	\begin{alignat}{4}
		&   \max_{x^p \in \lbrace 0,1\rbrace ^n}  && \sum_{i=1}^n v^p_ix^p_i+  \sum_{k=1, k\neq p}^m \sum_{i=1}^{n} c^p_{k,i} x^p_i x^k_i  \label{KPA:1a_gen}\\[0.4ex]
		&\mbox{\ \ s. t.}  &&\sum_{i=1}^n w_i^p x_i^p \leq W^p.       \label{KPA:1b_gen}
	\end{alignat}
	\label{KPA}%
\end{subequations}
The parameters of this game are integer (but are not required to be non-negative).
This model can describe situations where $m$ entities aim to decide in which of $n$ projects to invest such that each entity budget constraint \eqref{KPA:1b_gen} is satisfied and the associated payoffs are maximized \eqref{KPA:1a_gen}. The second summation in the payoff~\eqref{KPA:1a_gen} can describe a benefit, $c_{k,i}^p > 0$, or a penalization, $c_{k,i}^p <  0$, when both player $p$ and player $k$ invest in project $i$; note also that since all variables are binary $\left(x_{i}^p\right)^2 = x_{i}^p$; player $p$'s payoff function is linear in $x^p$. This means that in our algorithms, when we verify if a player has incentive to deviate from her current strategy, the variables $x^{-p}$ are fixed, and thus, the best reaction corresponds to an integer linear program.

In the knapsack game, each player $p$'s set of strategies $X^p$ is bounded, since she has at most $2^n$ feasible strategies. Therefore, by Theorem~\ref{lemma_finitelysupported}, it suffices to study finitely supported  equilibria. Since payoffs are linear, through the proof of Theorem~\ref{lemma_finitelysupported}, we deduce that the bound on the equilibria supports for each player is $n+1$. We can sightly improve this bound using  basic polyhedral theory (see \cite{Wolsey1988}).
First, note that a player $p$'s optimization problem is linear in her variables, implying  her set of pure optimal strategies to a fixed profile of strategies $\sigma^{-p} \in \Delta^{-p}$ to be in a facet of  $\conv(X^p)$. Second, the part in the payoffs of player $p$'s opponents that depends on player $p$'s strategy, only takes into account the expected value of $x^p$. The expected value of $x^p$ is a convex combination of player $p$'s pure strategies. Thus, putting together these two observations, when player $p$ selects an optimal mixed strategy $\sigma^p$ to $\sigma^{-p}$, the expected value of $x^p$ is in a facet of $\conv(X^p)$. A facet of $\conv(X^p)$ has dimension $n-1$, therefore, by Carath\'eodory's theorem~\citep{Bertsekas2003}, any point of this facet can be written as a convex combination of $n$ strategies of $X^p$. Thus,
\begin{lemma}
	Given an equilibrium $\sigma$ of the knapsack game, there is an equilibrium $\tau$ such that $\vert supp(\tau^p) \vert \leq n$ and $\Pi^p(\sigma)=\Pi^p(\tau)$,  for each $p=1,\ldots,m$.
	\label{lemma:support_size_knapsack_game}
\end{lemma}

\subsubsection{Two-player kidney exchange game} \label{sec:KEG_game}
\cite{Roth_Sonmez_Unver_2005_a} introduced a kidney exchange game between transplantation centers. We recover the 2-player kidney exchange game version of \citet{Carvalho2016}. In this game, there are two countries with their own kidney exchange programs (KEPs), \ie, a national system that allows patients in a need of a kidney transplant to register with an incompatible donor, and to perform donor exchanges. Mathematically, KEPs are represented by directed graphs where each vertex is an incompatible patient-donor pair and an arc from a vertex $v$ to a vertex $u$ means that the donor of pair $v$ is compatible with the patient of pair $u$. The goal of a KEP is to maximize the patients benefit by finding disjoint cycles in the graph, \ie, feasible kidney exchanges. In the kidney exchange game, countries try to optimize their patients benefit by joining their KEP pools. In this game, countries $A$  and $B$ aim to solve the following problems:
\begin{subequations}
	\begin{alignat}{5}
		(Country_A)&\maxim_{x^A \in \{ 0,1\}^{\vert  C^A  \vert+ \vert  I \vert} }  &&  \ \ \sum_{c \in C^A} w_c^A x_c^A + \sum_{c \in I} w_c^A x_c^A x_c^B\label{objCountryGeneral_A} \\[0.4ex]
		&\mbox{subject to~~}  &&\sum_{c \in C^A: i \in c} x_c^A+\sum_{c \in I: i \in c} x_c^A  \leq 1 \quad  \forall i \in V^A,      \label{PatientOneGeneral_A}
	\end{alignat}
	\label{NKEGproblem_A}
\end{subequations}
\begin{subequations}
	\begin{alignat}{5}
		(Country_B)&\maxim_{x^B \in \{ 0,1\}^{\vert  C^B  \vert+ \vert  I \vert} }  &&  \ \ \sum_{c \in C^B} w_c^B x_c^B + \sum_{c \in I} w_c^B x_c^A x_c^B\label{objCountryGeneral_B} \\[0.4ex]
		&\mbox{subject to~~}  &&\sum_{c \in C^B: i \in c} x_c^B+\sum_{c \in I: i \in c} x_c^B  \leq 1 \quad  \forall i \in V^B,      \label{PatientOneGeneral_B}
	\end{alignat}
	\label{NKEGproblem_B}%
\end{subequations}
where $V^p$ and $C^p$ are the set of incompatible patient-donor pairs and feasible cycles of country $p \in \{A,B\}$, $w_c^p$ is the benefit of the patients from country $p \in \{A,B\}$ in cycle $c$, and $I$ is the set of feasible international cycles. Essentially, countries directly decide their national-wide exchanges, while international exchanges require them both to agree. In~\cite{Carvalho2016}, the international exchange decisions are taken instead by an independent agent that maximizes the overall benefit over the available pairs for international exchanges:
\begin{subequations}
	\begin{alignat}{5}
	&\maxim_{y\in \{ 0,1\}^{\vert  I \vert} }  &&  \ \ \sum_{c \in I} (w_c^A+w_c^B) y_c\label{obj_IA} \\[0.4ex]
		&\mbox{subject to~~}  && \sum_{c \in I: i \in c} y_c  \leq 1 - \sum_{c \in C^A: i \in c} y_c^A -\sum_{c \in C^B: i \in c} y_c^B \quad  \forall i \in V^A  \cup V^B,    \label{PatientOneGeneral_IA}
	\end{alignat}
	\label{NKEG_IA}%
\end{subequations}
with $y^A \in \{0,1\}^{\vert  C^A  \vert}$ and  $y^B \in \{0,1\}^{\vert  C^B  \vert}$ corresponding to the internal exchanges selected by country A and B, respectively. Here, countries A and B simultaneously choose their internal exchanges, and afterwards the independent agent selects the international exchanges. In our {\IPG}~\eqref{NKEGproblem_A}-\eqref{NKEGproblem_B}, countries A and B simultaneously choose their internal exchanges and the international exchanges they would like to take place. Without this simplification in our {\IPG}, each country would have in its constraints the independent agent optimization, rendering even the computation of a player best reaction computationally hard~\citep{smeulders2020}. 
 In our setting, there is no independent agent decision. Instead, there is a direct agreement of the countries in the international exchanges (bilateral terms in the payoffs). In fact, we can prove that the obtained game contains the pure Nash equilibria of the original game in~\cite{Carvalho2016}.

\begin{lemma}
Any pure Nash equilibrium of~\cite{Carvalho2016} has an equivalent pure Nash equilibrium in the {\IPG} described by Problems~\eqref{NKEGproblem_A} and~\eqref{NKEGproblem_B}, in the sense that countries internal strategies coincide, as well as, their agreed international exchanges (thus,  payoffs are maintained). 
\label{lem:KEG}
\end{lemma} 

\begin{proof}
Let $(\hat{y}^A,\hat{y}^B,\hat{y})$ be a pure Nash equilibrium   of the game in~\cite{Carvalho2016}. We claim that $(\hat{x}^A,\hat{x}^B)$ with $\hat{x}^p = (\hat{y}^p,\hat{y})$ for $p \in \{A,B\}$ is a pure Nash equilibrium of the {\IPG} described by Problems~\eqref{NKEGproblem_A} and~\eqref{NKEGproblem_B}. To see this, we show that player A has no incentive to deviate from   $(\hat{x}^A,\hat{x}^B)$; for player B the reasoning is completely analogous. First, note that any deviation from the international exchanges in $y$ can only decrease player A's benefit as player B has agreed exactly on the international exchanges $\hat{y}$. Hence, we just need to consider deviations from the internal exchanges $\hat{y}^A$ (and eventually, replace some ones by zeros in $\hat{y}$ for international exchanges that become unavailable). However, if player A can increase the benefit of internal exchanges in the {\IPG}, this would have also been true in the game in~\cite{Carvalho2016}. 
\end{proof}

Three important remarks must be stressed. First, the result above does not hold for (general) mixed equilibria. However, as the experiments will show, we always determine pure equilibria which have the practical value of being simpler to implement. Second, the opposite direction of the lemma does not hold: if only cross-border exchanges exist, \ie, $C^A=C^B=\emptyset$, then  $(x^A,x^B)=(\textbf{0},\textbf{0})$ is an equilibrium of the {\IPG}, while in~\cite{Carvalho2016}, the independent agent would select at least one cross-borde exchange. Third, in the proof of Lemma~\ref{lem:KEG}, we did not use in our reasoning the length of the cycles in $C^A$ and $C^B$.  This is particularly interesting since in~\cite{Carvalho2016}, only the properties of the game when there are cycles of length 2 were characterized. Thus, our methodology allow us to go beyond this element; in practice, most countries consider cycles of length 2 and 3~\citep{BIRO2019}.

Although, our {\IPG} formulation of the kidney exchange game avoids loosing pure equilibria of the original game~\cite{Carvalho2016}, this game suffers from the existence of multiple pure Nash equilibria: for any player B's strategy $x^B$, there is a player A's best response $x^A$ where only a subset of the international exchanges  $c \in I$ with $x^B_c=1$ are selected; the same holds with the players' roles inverted; hence, $(x^A, \bar{x}^B)$ where $\bar{x}^B$ is player B's best response to $x^A$ is a Nash equilibrium. Motivated by this we decided to use the concept of maximal Nash equilibrium.

\begin{definition}
In an {\IPG} where all players variables are restricted to take binary values, a pure strategy $x^p$ for $p \in M$ is called maximal, if $(x^p_1,\ldots,x^p_{j-1},1,x^p_{j+1},\ldots,x^p_{n_p}) \notin X^p$ for $j=1,\ldots, n_p$ with $x^p_j=0$. A Nash equilibrium $\sigma \in \Delta$ is maximal if for each player $p \in M$, each $x^p \in \supp(\sigma^p)$ is maximal. 
\end{definition}


 \begin{lemma}
For the {\IPG} described by Problems~\eqref{NKEGproblem_A} and~\eqref{NKEGproblem_B}, any Nash equilibrium of it restricted to maximal strategies is a Nash equilibrium of the game without this restriction. Moreover, the pure Nash equilibria of~\cite{Carvalho2016} are contained on the equilibria of this restricted game.
\end{lemma}
\begin{proof}
Let $\sigma$ be a Nash equilibrium of the restricted {\IPG} game. If it is not a Nash equilibrium of the {\IPG}, then w.l.o.g. country A has incentive to deviate. This deviation must be to a non-maximal strategy $\bar{x}^A$. Note that we can make $\bar{x}^A$ maximal by  changing its 0 entries to 1 until the strategy becomes infeasible. Note that making $\bar{x}^A$ maximal does not damage on country $A$'s payoff. This contradicts the fact that $\sigma$ was a Nash equilibrium of the restricted game.

For the second part of the lemma, start by noting that in a pure equilibrium $(y^A,y^B)$ of~\cite{Carvalho2016}, each player $p$ is certainly selecting a maximal set of internal cycles, \ie, no entry with $y_c^p=0$ for $c \in C^p$ can become 1 without violating feasibilty. Furthermore, in~\cite{Carvalho2016}, there is an agent that maximizes the overall benefit of international exchanges once players have decided their internal cycles. Hence, this agent is also selecting a maximal set of international exchanges. Therefore, by Lemma~\ref{lem:KEG},  $(y^A,y^B)$ can be converted in an equilibrium of the {\IPG} described by Problems~\eqref{NKEGproblem_A} and~\eqref{NKEGproblem_B}.
\end{proof}

In this way, we restrict our experiments to maximal strategies (and consequently, maximal equilibria) for the kidney exchange game in an attempt to improve social welfare outcomes, avoiding \emph{dominated} equilibria. Before proceeding to the next section, we observe that in~\cite{Carvalho2016}, it was show that when cycles are restricted to length 2, the game is potential and it was conjectured a  potential function when the  cycles length is restricted to 3. In appendix~\ref{app:no_potential}, we show a negative answer to the conjectured function.

\subsubsection{Competitive lot-sizing game. } \label{sec:competitive_lotsizing_game}
The competitive lot-sizing game~\citep{CarvalhoTelhaVyve} is a Cournot competition played through $T$ periods by a set of firms (players) that produce the same good. Each firm has to plan its production as in the lot-sizing problems (see \cite{Pochet:2006}) but, instead of satisfying a known demand in each period of the time horizon, the demand depends on the total quantity of the produced good that exists in the market. Each firm $p$ has to decide how much will be produced in each time period $t$ (production variable $x^p_t$) and how much will be placed in the market (variable $q^p_t$). There are set-up and variable (linear) production costs, upper limit on production quantities, and a producer can build inventory (variable $h^p_t$) by producing in advance. In this way, we obtain the following model
for each firm $p$:
\begin{subequations}
	\begin{alignat}{5}
		&  \max_{y^p,x^p, q^p, h^p }  \ \     \sum_{t=1}^T (a_t-b_t \sum_{j=1}^m q_t^j)q_t^p - \sum_{t=1}^T F_t^p y_t^p  - \sum_{t=1}^T C^p_t x_t^p - \sum_{t=1}^T H^p_t h_t^p \label{ULSG_obj}\\[0.4ex]
		&\mbox{subject to~~}   x^p_t+h_{t-1}^p = h^p_t+q^p_t   \quad  \textrm{ for } t=1, \ldots , T \label{ULSG_stock} \\
		& \hspace{2.1cm}  0\leq x_t^p \leq M^p_ty_t^p   \hspace{1.1cm} \textrm{ for } t=1, \ldots , T    \label{ULSG_capacity} \\
		& \hspace{2.1cm}  h^p_0=h^p_T=0       \label{General_Initial_end} \\
		& \hspace{2.1cm}  y_t^p \in \lbrace 0,1 \rbrace    \hspace{1.6cm} \textrm{ for } t=1, \ldots , T    \label{ULSG_binary}
	\end{alignat}
	\label{ULSG}%
\end{subequations}
where $F_t^p$ is the set-up cost, $C^p_t$ is the variable cost, $H^p_t$ is the inventory cost and $M^p_t$ is the production capacity for period $t$; $a_t-b_t \sum_{j=1}^m q_t^j $ is the unit market price. The payoff function \eqref{ULSG_obj} is firm $p$'s total profit; constraints \eqref{ULSG_stock} model product conservation between periods; constraints \eqref{ULSG_capacity} ensure that the quantities produced are non-negative and whenever there is production  ($x^p_t > 0$), the binary variable $y^p_t$ is set to 1 implying the payment of the setup cost $F_t^p$. 

Each firm $p$'s payoff function~\eqref{ULSG_obj} is quadratic in $q^p$ due to the term $\sum_{t=1}^T -b_t(q_t^p)^2$. Next, we show that it satisfies the Lipschitz condition which guarantees that our algorithms compute an $\varepsilon$-equilibrium in finite time. Noting that player $p$ does not have incentive to select $q_t^p> \frac{a_t}{b_t}$ (since it would result in null market price), we get
\begin{subequations}
	\begin{alignat*}{4}
		&|\sum_{t=1}^T b_t(q_t^p)^2 - \sum_{t=1}^T b_t(\hat{q}_t^p)^2| & =&|\sum_{t=1}^T b_t \left( (q_t^p)^2 - (\hat{q}_t^p)^2 \right) | \\[0.4ex]
		&                                                                   & = & |\sum_{t=1}^T b_t \left( (q_t^p)+ (\hat{q}_t^p) \right) \left( (q_t^p)- (\hat{q}_t^p) \right) | \\
		&                                                                   & \leq & \sqrt{\sum_{t=1}^T b_t^2 \left( (q_t^p)+ (\hat{q}_t^p) \right)^2 } \sqrt{\sum_{t=1}^T \left( (q_t^p)- (\hat{q}_t^p) \right)^2 } \\
		&                                                                   &  \leq & \sqrt{\sum_{t=1}^T b_t^2 \left( \frac{2a_t}{b_t}\right)^2 }\cdot \parallel q^p -\hat{q}^p \parallel \\
		&                                                                   &  \leq & \sqrt{\sum_{t=1}^T 4a_t^2 }\cdot \parallel q^p -\hat{q}^p \parallel.
	\end{alignat*}
\end{subequations}
In the third step, we used Cauchy--Schwarz inequality. In the fourth inequality, we use the upper  bound $ \frac{a_t}{b_t}$ on the quantities placed in the market.

In~\cite{CarvalhoTelhaVyve}, it was proven that there is a function that is \textit{potential} for this game; a maximum of this function is a (pure) equilibrium (recall Lemma~\ref{lem:Monderer_Shapley}). This is an additional motivation to analyze our framework in this problem: it can be compared with the maximization of the associated potential function. 


\subsection{Implementation details} \label{subsec:scpecialized_functions}
Both our implementations of the {\modSGM}  and {\SGM} use the following specialized functions.

\paragraph{$Initialization(IPG)$.} The equilibrium computed by our methods depends on their initialization as the following example illustrates. 

\begin{example}
	Consider an instance of the two-player competitive lot-sizing game with the following parameters: $T=1$, $a_1=15$, $b_1=1$, $M_1^1=M_1^2=15$, $C^1_1=C^2_1=H^1_1=H^2_1=0$, $F^1_1=F_1^2=15$. It is a one-period game, therefore the inventory variables, $h^1_1$ and $h^2_1$, can be removed and the quantities produced are equal to the quantities placed in the market (that is, $x^1_1=q^1_1$ and $x^2_1=q^2_1$). Given the simplicity of the players optimization programs~\eqref{ULSG}, we can analytically compute the players' best reactions that are depicted in Figure~\ref{best_reactions_example2}.
	
	The game possesses two (pure) equilibria: $\hat{x}=(\hat{x}^1,\hat{y}^1,\hat{x}^2,\hat{y}^2)=(0,0;7.5,1)$ and $\tilde{x}=(\tilde{x}^1,\tilde{y}^1,\tilde{x}^2,\tilde{y}^2)=(7.5,1;0,0)$. Thus, depending on the initialization of {\modSGM}, it will terminate with $\hat{x}$ or $\tilde{x}$: Figure~\ref{best_reactions_example2} depicts the convergence to $\hat{x}$ when the initial sampled game is $\mathbb{S}=\lbrace (2,1) \rbrace \times \lbrace (5,1) \rbrace$ and to $\tilde{x}$ when the initial sampled game is $\mathbb{S}=\lbrace (4,1) \rbrace \times \lbrace (1,1) \rbrace$. 
	
	\begin{figure}[th]
		\begin{center}
		\includegraphics[scale=0.4]{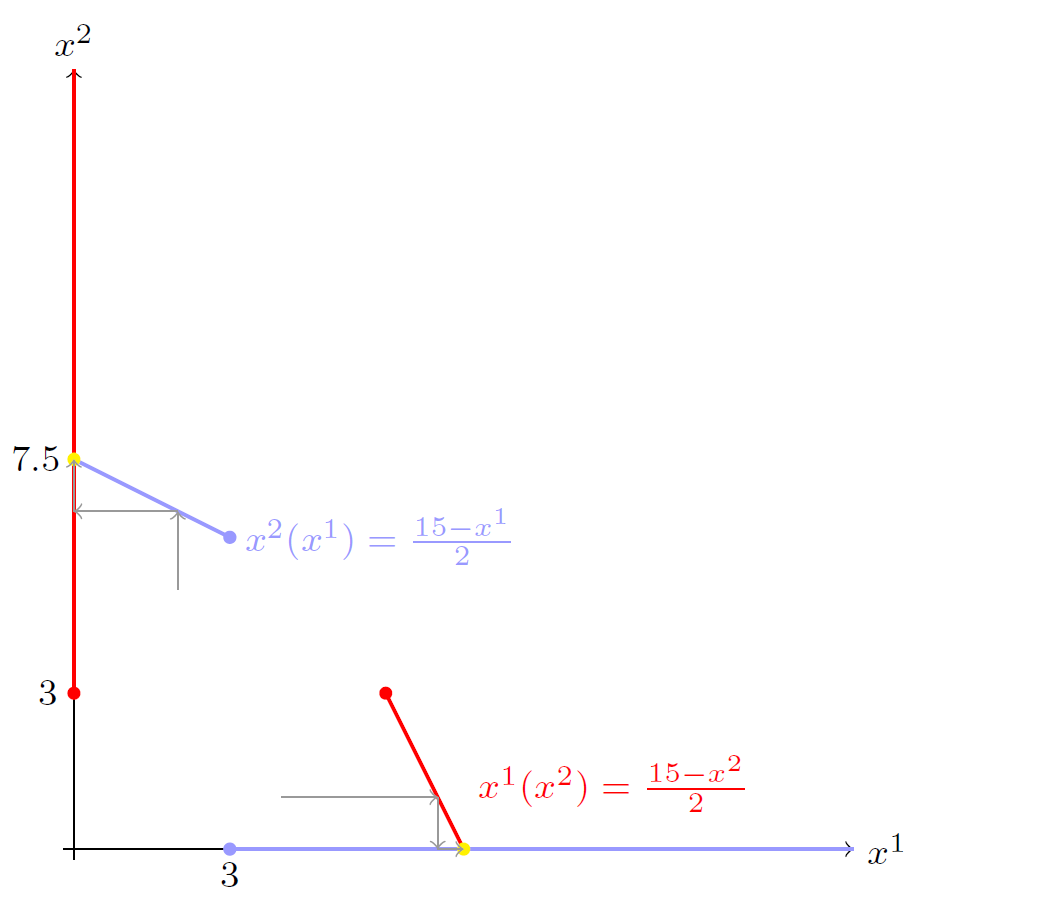}
		\end{center}
		\caption{Players' best reaction functions. Pure {\NE} are indicated in yellow. The gray arrows display the sampled game evolution for the initial games $\mathbb{S}=\lbrace (2,1) \rbrace \times \lbrace (5,1) \rbrace$ and  $\mathbb{S}=\lbrace (4,1) \rbrace \times \lbrace (1,1) \rbrace$. }
		\label{best_reactions_example2}
	\end{figure}
\end{example}

In an attempt to keep as small as possible the size of the sampled games (\ie, number of strategies explicitly enumerated), the initialization implemented computes a unique pure strategy for each player. We experimented initializing the algorithm with the social optimal strategies (strategies that maximize the total players' payoffs), pure equilibrium for the potential part of the game\footnote{We only experimented this for the knapsack game. We consider the potential part of the knapsack game, when the parameters $c^p_{k,i}$ in each player's payoff function are replaced by $\frac{1}{2} (c^p_{k,i}+c^k_{p,i})$ in player $p$'s payoff.}, and optimal strategies if the players were alone in the game (\ie, opponents' variables were set to be zero). In general, there was no evident advantage on the speed of computations for one of these initializations. This result was somehow expected, since, particularly for the knapsack game instances, it is not evident whether the game has an important coordination part (in the direction of social optimum) or an important conflict part. Therefore, our implementation initializes with the players' strategies that are optimal when they are alone in the game for the knapsack and lot-sizing game. For the kidney exchange game, the initialization does not change significantly the speed of equilibria computation but it interferes in the equilibria found which in this context is associated with the set of patients expected to receive a transplant. For this reason, the kidney exchange game is initialized with each country optimal strategy when it controls the opponents variables. This allows countries to select their preferred international exchanges, allowing them to take advantage of the joint KEP. Otherwise, if we keep the same initialization of the knapsack and lot-sizing games, we will see {\NE} with lower social welfare (\ie, total benefit for patients).

\paragraph{$PlayerOrder(\mathbb{S}_{dev_0}, \ldots,\mathbb{S}_{dev_k})$.} The equilibrium returned by our algorithms depends on the players' order when we check their incentives to deviate in the Termination steps: for the equilibrium $\sigma_k$ of the sampled game $k$, there might be more than one player with incentive to deviate from $\sigma_k$, thus the successor will depend on the player that is selected. If players' index order is considered, the algorithm may take longer to converge to an equilibrium: it will be likely that it first finds an equilibrium of the game restricted to players 1 and 2, then an equilibrium of the game restricted to players 1, 2 and 3, and so on. Thus, this implementation sorts the players by decreasing order of number of previous iterations without receiving a new strategy.

\paragraph{$DeviationReaction(p,\sigma_k^{-p},\Pi^p (\sigma_k),\varepsilon,IPG)$.} When checking if a player $p$ has incentive to deviate, it suffices to determine whether she has a strategy that strictly increases her payoff when she unilaterally deviates to it. Nowadays, there are software tools that can solve mixed integer linear and  quadratic programming problems\footnote{In the knapsack and kidney exchange games, a player's best reaction problem is an integer linear programming problem. In the competitive lot-sizing problem, a player best reaction problem is a mixed integer quadratic programming problem (it becomes continuous and concave once the binary variables $y^p$ are fixed).} effectively. Thus, our implementation solves the player $p$'s best reaction problem~\eqref{GeneralProblem} to $\sigma^{-p}_k$. We use Gurobi 9.0.0\footnote{\url{www.gurobi.com}} to solve these reaction problems.

\paragraph{$SortSizes(\sigma_0,\ldots,\sigma_{k-1})$.} \cite{Porter2008642} recommend that the support strategies' enumeration starts with support sizes ordered, first, by total size ($\sum_{p=1}^m s^p$ with $s^p$ the support size for player $p$), and, second, by a measure of balance (except, in case of a 2-players game where these criteria importance is reversed). However, in our methods, from one sampled game to its successor or predecessor, the sampled game at hand just changes by one strategy, and thus we expect that the equilibria will not change too much either (in particular, the support sizes of consecutive sampled games are expected to be close). Therefore, 
our criterion to sort the support sizes is by increasing order of:
\begin{description}
	\item[For $m = 2$:]  first, balance, second, maximum player's support size distance to the one of the previously computed equilibrium,  third, maximum player's support size distance to the one of the previously computed equilibrium  plus 1  and, fourth, sum of the players' support sizes; 
	
	\item[For $m \geq 3$:] first, maximum player's support size distance to the one of the previously computed equilibrium, second, maximum player's support size  distance to the one of the previously computed equilibrium plus 1, third, sum of the players' support sizes and, fourth, balance
	.
\end{description}
For the initial sampled game, the criteria coincide with {\PNS}.

\paragraph{$SortStrategies(\mathbb{S}, \sigma_0,\ldots,\sigma_{k-1})$.} Following the previous reasoning, the strategies of the current sampled game are sorted  by decreasing order of their probability in the predecessor equilibrium. Thus, the algorithm will prioritize finding equilibria using the support strategies of the predecessor equilibrium.

Note that the function {\PNS}$_{adaptation}(\mathbb{S},x(k),\mathbb{S}_{dev_{k+1}},Sizes_{ord},Strategies_{ord})$ is specific for the  {\modSGM}. The basic {\SGM} calls {\PNS} without any requirement on the strategies that must be in the support of the next equilibrium to be computed; in other words, $x(k)$ and $\mathbb{S}_{dev_{k+1}}$ are not in the input of the {\PNS}.

\subsection{Computational results} \label{subsec:computational_results}

In this section, we will present the computational results for the application of the modified {\SGM} and {\SGM} to the knapsack, kidney exchange and competitive lot-sizing games  in order to define a benchmark and to validate the importance of the modifications introduced. For the competitive lot-sizing game, we further compare these two methods with the maximization of the game's potential function (which corresponds to a pure equilibrium). In our computational analyzes, we also include the {\SGM} adaption for the computation of {\CE}. 

For building the games's data, we have used the Python's random module; see \cite{PythonRandom}.  
 All algorithms have been coded in Python 3.8.3. Since for our three {\IPG}s the Feasibility Problems are linear (due to the bilateral interaction of the players in each of their objective functions), we use Gurobi 9.0.0 to solve them. The experiments were conducted on a  Intel Xeon Gold 6226 CPU processor at 2.70 GHz, running under Oracle Linux Server 7.9, and restricted to a single CPU thread (with exception to Gurobi calls which were restricted to at most 2 CPU threads).

\subsubsection{Knapsack Game}  In these computations, the value of $\varepsilon$ wasset to  zero since this is a purely integer programming game.
The parameters $v^p_i$, $c^p_{k,i}$, and $w^p_i$ are drawn independently from a uniform distribution  in the interval $\left[ -100, 100 \right] \cap \mathbb{Z}$. For each value of the pair $(n,m)$, 10 independent instances were generated. The budget $W^p$ is set to $\lfloor \frac{\textrm{INS}}{11}\sum_{i=1}^n w^p_i \rfloor$ for the instance number ``INS''.  

\paragraph{{\NE} computation. } Tables~\ref{table_fastNE_KnapsackGame_new} and~\ref{table_fastNE_KnapsackGame3} report the results of {\modSGM} and {\SGM} algorithms. The tables show the number of items ($n$), the instance identifier (``INS''),  the CPU time in seconds (``time''), the number of sampled games (``iter''), the type of equilibrium computed, pure (``pNE'') or strictly mixed (``mNE''), in the last case, the support size of the {\NE} is reported, the number of strategies in the last sampled game ($ \prod_{p=1}^m \vert \mathbb{S}^p \vert $) and the number of backtrackings (``numb. back''). We further report the average results for each set of instances of size $n$. The algorithms had a limit of one hour to solve each instance. Runs with ``tl'' in the column time indicate the cases where algorithms reached the time limit. In such cases, the support size of the last sampled game's equilibrium is reported and we do not consider those instances in the average results row.

As the instance size grows, both in the size $n$ and in the number of players $m$, the results make evident the advantage of the {\modSGM}. Since a backward step is unlikely to take place and the number of sampled games is usually equal for both algorithms, the advantage is in the support enumeration: {\modSGM} reduces the support enumeration space by imposing at iteration $k$ the strategy $x(k)$ to be in the support of the equilibrium, while {\SGM} does not. Later in this section, we discuss the reasons why backtracking is unlikely to occur.

In Table~\ref{table_fastNE_KnapsackGame_new}, we can observe that for instance 6 with $n=100$, the  {\modSGM} computational time is significantly higher than {\SGM}. This atypical case is due to the fact that both algorithms have different support enumeration priorities, and therefore, they compute the same equilibria on their initial iterations, but at some point, the algorithms may determine different equilibria, leading to different successor sampled games. Nevertheless, for this instances, {\modSGM} and {\SGM} output the same {\NE}.

We note that the bound $n$ for the players' support sizes in an equilibrium (recall Lemma~\ref{lemma:support_size_knapsack_game}) did not contribute to prune the search space of {\PNS} support enumeration, since the algorithm terminates with sampled games ofmuch  smaller size.

\bgroup
\def\arraystretch{1.5}
\begin{table}[ptbh!]
	\centering
	\tiny
	\tabcolsep=2pt
	\begin{tabular}{lr|crccccc|crccccc}
		&      & \multicolumn{7}{c|}{{\modSGM}}    & \multicolumn{7}{c}{{\SGM}}  \\[0.1cm]
		\cline{1-2} \cline{2-10} \cline{10-16}
		$n$ & INS    & time      & iter     &  pNE & mNE   & $ \prod_{p=1}^m \vert \mathbb{S}^p \vert $    & numb. back      & & time & iter & pNE & mNE &  $ \prod_{p=1}^m \vert \mathbb{S}^p \vert $ \\ 
20 &0 & 0.03 & 10 & 0 & [3, 3] & [5, 6] & 0 & &0.03 & 10 & 0 & [3, 3] & [5, 6]& \\ 
&1 & 0.02 & 8 & 0 & [3, 3] & [4, 5] & 0 & &0.02 & 8 & 0 & [3, 3] & [4, 5]& \\ 
&2 & 0.01 & 4 & 0 & [2, 2] & [2, 3] & 0 & &0.01 & 4 & 0 & [2, 2] & [2, 3]& \\ 
&3 & 0.01 & 6 & 0 & [2, 2] & [4, 3] & 0 & &0.01 & 6 & 0 & [2, 2] & [4, 3]& \\ 
&4 & 0.01 & 4 & 1 & 0 &[3, 2] & 0 & & 0.01 & 4 & 1 & 0 &[3, 2]& \\ 
&5 & 0.01 & 4 & 0 & [2, 2] & [3, 2] & 0 & &0.01 & 4 & 0 & [2, 2] & [3, 2]& \\ 
&6 & 0.01 & 4 & 1 & 0 &[3, 2] & 0 & & 0.01 & 4 & 1 & 0 &[3, 2]& \\ 
&7 & 0.01 & 5 & 0 & [2, 2] & [3, 3] & 0 & &0.01 & 5 & 0 & [2, 2] & [3, 3]& \\ 
&8 & 0.01 & 5 & 0 & [2, 2] & [3, 3] & 0 & &0.01 & 5 & 0 & [2, 2] & [3, 3]& \\ 
&9 & 0.01 & 4 & 1 & 0 &[3, 2] & 0 & & 0.01 & 4 & 1 & 0 &[3, 2]& \\ 
		&     & time    & iter      & pNE & mNE & $\vert S^1 \vert$       & $\vert S^2 \vert$  &  & time    & iter      & pNE     & mNE   & $\vert S^1 \vert$     & $\vert S^2 \vert$ &  \\ \cline{3-16} 
 & avg  & 0.01 & 5.40 & 0.3 & 0.7 & 3.30 & 3.10 & &0.01 & 5.40 & 0.3 & 0.7 & 3.30 & 3.10 \\ \hline
 	$n$ & INS    & time      & iter     &  pNE & mNE   & $ \prod_{p=1}^m \vert \mathbb{S}^p \vert $    & numb. back      & & time & iter & pNE & mNE &  $ \prod_{p=1}^m \vert \mathbb{S}^p \vert $ \\ 
40 &0 & 1.02 & 18 & 0 & [6, 6] & [10, 9] & 0 & &1.48 & 18 & 0 & [6, 6] & [10, 9]& \\ 
&1 & 1.83 & 17 & 0 & [4, 4] & [8, 10] & 0 & &3.54 & 17 & 0 & [4, 4] & [8, 10]& \\ 
&2 & 0.07 & 14 & 0 & [4, 4] & [7, 8] & 0 & &0.07 & 14 & 0 & [4, 4] & [7, 8]& \\ 
&3 & 0.18 & 16 & 0 & [5, 5] & [10, 7] & 0 & &0.27 & 16 & 0 & [5, 5] & [10, 7]& \\ 
&4 & 0.01 & 5 & 0 & [2, 2] & [3, 3] & 0 & &0.01 & 5 & 0 & [2, 2] & [3, 3]& \\ 
&5 & 0.01 & 5 & 0 & [2, 2] & [3, 3] & 0 & &0.01 & 5 & 0 & [2, 2] & [3, 3]& \\ 
&6 & 0.02 & 7 & 0 & [2, 2] & [4, 4] & 0 & &0.02 & 7 & 0 & [2, 2] & [4, 4]& \\ 
&7 & 0.98 & 20 & 0 & [6, 6] & [9, 12] & 0 & &1.78 & 20 & 0 & [6, 6] & [9, 12]& \\ 
&8 & 0.25 & 15 & 0 & [5, 5] & [8, 8] & 0 & &0.41 & 15 & 0 & [5, 5] & [8, 8]& \\ 
&9 & 11.70 & 20 & 0 & [5, 5] & [10, 11] & 0 & &21.99 & 20 & 0 & [5, 5] & [10, 11]& \\ 
		&     & time    & iter      & pNE & mNE & $\vert S^1 \vert$       & $\vert S^2 \vert$  &  & time    & iter      & pNE     & mNE   & $\vert S^1 \vert$     & $\vert S^2 \vert$ &  \\ \cline{3-16} 
& avg  & 1.61 & 13.70 & 0.0 & 1.0 & 7.20 & 7.50 & &2.96 & 13.70 & 0.0 & 1.0 & 7.20 & 7.50 \\ 
\hline
$n$ & INS    & time      & iter     &  pNE & mNE   & $ \prod_{p=1}^m \vert \mathbb{S}^p \vert $    & numb. back      & & time & iter & pNE & mNE &  $ \prod_{p=1}^m \vert \mathbb{S}^p \vert $ \\
80&0 & 0.56 & 18 & 0 & [5, 5] & [9, 10] & 0 & &0.78 & 18 & 0 & [5, 5] & [9, 10]& \\ 
&1 & 0.79 & 17 & 0 & [8, 8] & [9, 9] & 0 & &1.12 & 17 & 0 & [8, 8] & [9, 9]& \\ 
&2 & 0.18 & 14 & 0 & [5, 5] & [8, 7] & 0 & &0.26 & 14 & 0 & [5, 5] & [8, 7]& \\ 
&3 & 10.55 & 22 & 0 & [6, 6] & [11, 12] & 0 & &21.17 & 22 & 0 & [6, 6] & [11, 12]& \\ 
&4 & tl & 30 & 0 & [8, 9] & [15, 16] & 0 & &tl & 29 & 0 & [9, 9] & [15, 15]& \\ 
&5 & 7.21 & 20 & 0 & [7, 7] & [11, 10] & 0 & &12.25 & 20 & 0 & [7, 7] & [11, 10]& \\ 
&6 & 5.04 & 19 & 0 & [7, 7] & [10, 10] & 0 & &8.56 & 19 & 0 & [7, 7] & [10, 10]& \\ 
&7 & 0.04 & 9 & 0 & [4, 4] & [5, 5] & 0 & &0.04 & 9 & 0 & [4, 4] & [5, 5]& \\ 
&8 & 2.14 & 18 & 0 & [7, 7] & [10, 9] & 0 & &3.19 & 18 & 0 & [7, 7] & [10, 9]& \\ 
&9 & 24.65 & 27 & 0 & [5, 5] & [14, 14] & 0 & &58.41 & 27 & 0 & [5, 5] & [14, 14]& \\ 
		&     & time    & iter      & pNE & mNE & $\vert S^1 \vert$       & $\vert S^2 \vert$  &  & time    & iter      & pNE     & mNE   & $\vert S^1 \vert$     & $\vert S^2 \vert$ &  \\ \cline{3-16} 
 & avg  & 5.69 & 18.22 & 0.0 & 0.9 & 9.67 & 9.56 & &11.75 & 18.22 & 0.0 & 0.9 & 9.67 & 9.56 \\ 
\hline
 $n$ & INS    & time      & iter     &  pNE & mNE   & $ \prod_{p=1}^m \vert \mathbb{S}^p \vert $    & numb. back      & & time & iter & pNE & mNE &  $ \prod_{p=1}^m \vert \mathbb{S}^p \vert $ \\
 100 &0 & 0.12 & 14 & 0 & [5, 5] & [8, 7] & 0 & &0.14 & 14 & 0 & [5, 5] & [8, 7]& \\ 
 &1 & 0.09 & 11 & 0 & [5, 5] & [6, 6] & 0 & &0.10 & 11 & 0 & [5, 5] & [6, 6]& \\ 
 &2 & 62.70 & 25 & 0 & [7, 7] & [15, 11] & 0 & &108.07 & 25 & 0 & [7, 7] & [15, 11]& \\ 
 &3 & 0.23 & 14 & 0 & [5, 5] & [8, 7] & 0 & &0.33 & 14 & 0 & [5, 5] & [8, 7]& \\ 
 &4 & 0.46 & 17 & 0 & [5, 5] & [9, 9] & 0 & &0.69 & 17 & 0 & [5, 5] & [9, 9]& \\ 
 &5 & tl & 28 & 0 & [9, 9] & [15, 14] & 0 & &tl & 27 & 0 & [7, 7] & [14, 14]& \\ 
 &6 & 1297.12 & 27 & 0 & [7, 8] & [13, 15] & 0 & &142.88 & 29 & 0 & [7, 8] & [12, 18]& \\ 
 &7 & 94.07 & 27 & 0 & [7, 7] & [12, 16] & 0 & &176.02 & 27 & 0 & [7, 7] & [12, 16]& \\ 
 &8 & 7.90 & 20 & 0 & [6, 6] & [10, 11] & 0 & &15.06 & 20 & 0 & [6, 6] & [10, 11]& \\ 
 &9 & tl & 28 & 0 & [8, 8] & [15, 14] & 0 & &tl & 28 & 0 & [8, 8] & [15, 14]& \\
 		&     & time    & iter      & pNE & mNE & $\vert S^1 \vert$       & $\vert S^2 \vert$  &  & time    & iter      & pNE     & mNE   & $\vert S^1 \vert$     & $\vert S^2 \vert$ &  \\ \cline{3-16} 
  & avg  & 182.84 & 19.38 & 0.0 & 0.8 & 10.12 & 10.25 & &55.41 & 19.62 & 0.0 & 0.8 & 10.00 & 10.62 \\ 
	\end{tabular}
\caption{Computational results for the determination of {\NE} on the e \textbf{knapsack game} with $m=2$.}
\label{table_fastNE_KnapsackGame_new}
\end{table}

\begin{table}[ptbh!]
	\centering
	\tiny
	\tabcolsep=2pt
	\begin{tabular}{lr|crccccc|crccccc}
		&      & \multicolumn{7}{c|}{{\modSGM}}    & \multicolumn{7}{c}{{\SGM}}  \\
		\cline{1-2} \cline{2-10} \cline{10-16}
		$n$ & INS    & time      & iter     &  pNE & mNE   & $ \prod_{p=1}^m \vert \mathbb{S}^p \vert $    & numb. back      & & time & iter & pNE & mNE &  $ \prod_{p=1}^m \vert \mathbb{S}^p \vert $ \\
10 &0 & 0.01 & 5 & 0 & [2, 1, 2] & [3, 2, 2] & 0 & &0.01 & 5 & 0 & [2, 1, 2] & [3, 2, 2]& & \\ 
&1 & 0.02 & 7 & 0 & [1, 2, 2] & [2, 3, 4] & 0 & &0.02 & 7 & 0 & [1, 2, 2] & [2, 3, 4]& & \\ 
&2 & 11.32 & 19 & 0 & [4, 3, 5] & [6, 7, 8] & 0 & &17.32 & 19 & 0 & [4, 3, 5] & [6, 7, 8]& & \\ 
&3 & 0.07 & 16 & 0 & [3, 2, 2] & [6, 6, 6] & 0 & &0.13 & 16 & 0 & [3, 2, 2] & [6, 6, 6]& & \\ 
&4 & 0.01 & 4 & 1 & 0 &[2, 2, 2] & 0 & & 0.01 & 4 & 1 & 0 &[2, 2, 2]& & \\ 
&5 & 0.01 & 6 & 0 & [2, 2, 1] & [4, 3, 1] & 0 & &0.01 & 6 & 0 & [2, 2, 1] & [4, 3, 1]& & \\ 
&6 & 0.01 & 7 & 1 & 0 &[3, 3, 3] & 0 & & 0.01 & 7 & 1 & 0 &[3, 3, 3]& & \\ 
&7 & 0.01 & 4 & 1 & 0 &[2, 2, 2] & 0 & & 0.01 & 4 & 1 & 0 &[2, 2, 2]& & \\ 
&8 & 21.44 & 22 & 0 & [2, 3, 2] & [8, 10, 6] & 0 & &33.07 & 22 & 0 & [2, 3, 2] & [8, 10, 6]& & \\ 
&9 & 0.02 & 10 & 0 & [2, 1, 2] & [5, 3, 4] & 0 & &0.02 & 10 & 0 & [2, 1, 2] & [5, 3, 4]& & \\ 	
		&     & time    & iter &  pNE   & mNE      & $\vert S^1 \vert$     & $\vert S^2 \vert$ &  $\vert S^3 \vert$    & time    & iter   & pNE   & mNE    & $\vert S^1 \vert$     & $\vert S^2 \vert$ & $\vert S^3 \vert$   \\ \cline{3-16}
 & avg  & 3.29 & 10.00 & 0.3 & 0.7 & 4.10 & 4.10 & 4.50  &5.06 & 10.00 & 0.3 & 0.7 & 4.10 & 4.10 & 4.50 \\ \hline
 	$n$ & INS    & time      & iter     &  pNE & mNE   & $ \prod_{p=1}^m \vert \mathbb{S}^p \vert $    & numb. back      & & time & iter & pNE & mNE &  $ \prod_{p=1}^m \vert \mathbb{S}^p \vert $ \\
20 &0 & 0.02 & 8 & 0 & [2, 1, 2] & [4, 2, 4] & 0 & &0.02 & 8 & 0 & [2, 1, 2] & [4, 2, 4]& & \\ 
&1 & 0.04 & 10 & 0 & [2, 2, 3] & [4, 4, 4] & 0 & &0.04 & 10 & 0 & [2, 2, 3] & [4, 4, 4]& & \\ 
&2 & 0.01 & 4 & 1 & 0 &[2, 2, 2] & 0 & & 0.01 & 4 & 1 & 0 &[2, 2, 2]& & \\ 
&3 & 0.02 & 8 & 0 & [1, 2, 2] & [3, 4, 3] & 0 & &0.02 & 8 & 0 & [1, 2, 2] & [3, 4, 3]& & \\ 
&4 & 0.10 & 13 & 0 & [2, 3, 2] & [5, 5, 5] & 0 & &0.15 & 13 & 0 & [2, 3, 2] & [5, 5, 5]& & \\ 
&5 & 0.02 & 8 & 1 & 0 &[3, 4, 3] & 0 & & 0.02 & 8 & 1 & 0 &[3, 4, 3]& & \\ 
&6 & 1.72 & 17 & 0 & [4, 3, 2] & [7, 7, 5] & 0 & &3.94 & 17 & 0 & [4, 3, 2] & [7, 7, 5]& & \\ 
&7 & 0.07 & 13 & 0 & [2, 2, 3] & [5, 5, 5] & 0 & &0.08 & 13 & 0 & [2, 2, 3] & [5, 5, 5]& & \\ 
&8 & 0.08 & 13 & 0 & [3, 4, 3] & [6, 5, 4] & 0 & &0.09 & 13 & 0 & [3, 4, 3] & [6, 5, 4]& & \\ 
&9 & 0.01 & 5 & 1 & 0 &[3, 2, 2] & 0 & & 0.01 & 5 & 1 & 0 &[3, 2, 2]& & \\  
	&     & time    & iter &  pNE   & mNE      & $\vert S^1 \vert$     & $\vert S^2 \vert$ &  $\vert S^3 \vert$    & time    & iter   & pNE   & mNE    & $\vert S^1 \vert$     & $\vert S^2 \vert$ & $\vert S^3 \vert$   \\ \cline{3-16}
 & avg  & 0.21 & 9.90 & 0.3 & 0.7 & 4.20 & 4.00 & 4.20  &0.44 & 9.90 & 0.3 & 0.7 & 4.20 & 4.00 & 4.20 \\ \hline
 	$n$ & INS    & time      & iter     &  pNE & mNE   & $ \prod_{p=1}^m \vert \mathbb{S}^p \vert $    & numb. back      & & time & iter & pNE & mNE &  $ \prod_{p=1}^m \vert \mathbb{S}^p \vert $ \\
40 &0 & 0.12 & 15 & 0 & [4, 2, 3] & [7, 3, 7] & 0 & &0.14 & 15 & 0 & [4, 2, 3] & [7, 3, 7]& & \\ 
&1 & 0.77 & 19 & 0 & [4, 4, 1] & [8, 9, 4] & 0 & &1.02 & 19 & 0 & [4, 4, 1] & [8, 9, 4]& & \\ 
&2 & 0.48 & 16 & 0 & [1, 2, 3] & [6, 6, 6] & 0 & &1.45 & 19 & 0 & [3, 3, 3] & [6, 8, 7]& & \\ 
&3 & 0.07 & 13 & 0 & [2, 3, 2] & [4, 6, 5] & 0 & &0.07 & 13 & 0 & [2, 3, 2] & [4, 6, 5]& & \\ 
&4 & tl & 28 & 0 & [4, 4, 4] & [10, 10, 10] & 0 & &tl & 28 & 0 & [4, 4, 4] & [10, 10, 10]& & \\ 
&5 & 2.81 & 23 & 0 & [4, 3, 4] & [9, 8, 8] & 0 & &4.47 & 23 & 0 & [4, 3, 4] & [9, 8, 8]& & \\ 
&6 & 0.22 & 16 & 0 & [2, 3, 2] & [7, 7, 4] & 0 & &0.30 & 16 & 0 & [2, 3, 2] & [7, 7, 4]& & \\ 
&7 & 11.25 & 24 & 0 & [4, 5, 2] & [10, 10, 6] & 0 & &27.52 & 24 & 0 & [4, 5, 2] & [10, 10, 6]& & \\ 
&8 & 1.67 & 18 & 0 & [4, 2, 3] & [7, 4, 9] & 0 & &2.85 & 18 & 0 & [4, 2, 3] & [7, 4, 9]& & \\ 
&9 & tl & 29 & 0 & [5, 4, 4] & [12, 9, 10] & 0 & &tl & 28 & 0 & [5, 4, 3] & [11, 9, 10]& & \\ 
	&     & time    & iter &  pNE   & mNE      & $\vert S^1 \vert$     & $\vert S^2 \vert$ &  $\vert S^3 \vert$    & time    & iter   & pNE   & mNE    & $\vert S^1 \vert$     & $\vert S^2 \vert$ & $\vert S^3 \vert$   \\ \cline{3-16}
 & avg  & 2.17 & 18.00 & 0.0 & 0.8 & 7.25 & 6.62 & 7.75  &4.73 & 18.38 & 0.0 & 0.8 & 7.25 & 6.88 & 8.00 \\ 
	\end{tabular}
	\caption{Computational results for the determination of {\NE} on the  \textbf{knapsack game} with $m=3$.}
	\label{table_fastNE_KnapsackGame3}
\end{table}
\bgroup
\def\arraystretch{1}

\paragraph{{\CE} computation. } Next, we present the computational results when the scheme of {\SGM} is adapted for the determination of {\CE} as described in Section~\ref{sec:extensions}. 

Tables~\ref{table_CE_KnapsackGame_2} and~\ref{table_CE_KnapsackGame_3} summarize our experiments. The columns meaning is the same as before. The new column ``$\tau$-based {\NE}?'' answers whether the computed correlated equilibrium $\tau$ leads to a {\NE}; recall Definition~\ref{def:ce_based_nash}. And column ``$1-\frac{Social(\sigma)}{Social(\tau)}$'' provides the social welfare decrease by moving from the correlated equilibrium $\tau$ to the {\NE} $\sigma$ computed in the previous experiment. We do not provide this column for the 2-player case because it is always 0.

For the 2-player case, see Table~\ref{table_CE_KnapsackGame_2}, the computation of {\CE} is much faster than the computation of {\NE} (recall Table~\ref{table_fastNE_KnapsackGame_new}). In fact, although the number of iterations for the computation of {\CE} is larger, it compensates the fact that searching for a {\CE} is much faster than searching for a {\NE} of the sampled games. The most surprising conclusion is on the fact that all computed {\CE} allowed the computation of {\NE}. This stresses further interest on {\CE} search.

In the 3-player case, see Table~\ref{table_CE_KnapsackGame_3}, conclusions are similar: the computation of {\CE} is generally much faster than the computation of {\NE} (an outlier  is instance 7 with $n=4$) and the number of iterations is larger for {\CE} determination. On the other hand, not all {\CE} allowed to determine a {\NE} accordingly with Definition~\ref{def:ce_based_nash}. Moreover, in this case, the selection of a {\CE} in the sampled game that optimizes social welfare seems to payoff in comparison with the social welfare of the previous computed {\NE}.

\bgroup
\def\arraystretch{1.5}
\begin{table}[ptbh!]
	\centering
	\tiny
	\tabcolsep=2pt
	\begin{tabular}{lr|crccccc}
		&         & \multicolumn{7}{c}{{\SGM}}  \\[0.1cm]
		\cline{1-8}
		$n$ & INS    & time & iter & $\tau$-based {\NE}? & $|\supp({\tau})|$ & $ \prod_{p=1}^m \vert \mathbb{S}^p \vert $ & & \\ 
20& 0 & 0.05 & 10 & YES & 9 & [5, 6] & & \\
 & 1 & 0.04 & 8 & YES & 9 & [4, 5] & & \\
 & 2 & 0.02 & 4 & YES & 4 & [2, 3] &  & \\
 & 3 & 0.02 & 6 & YES & 4 & [4, 3] & & \\
 & 4 & 0.01 & 4 & YES & 1 & [3, 2] &  & \\
 & 5 & 0.01 & 4 & YES & 4 & [3, 2] &  & \\
 & 6 & 0.01 & 4 & YES & 1 & [3, 2] & & \\
 & 7 & 0.02 & 5 & YES & 4 & [3, 3] &  & \\
 & 8 & 0.02 & 5 & YES & 4 & [3, 3] & & \\
 & 9 & 0.01 & 4 & YES & 1 & [3, 2] & & \\
 	&     & time    & iter      & $\tau$-based {\NE} & $|\supp({\tau})|$ & $\vert S^1 \vert$       & $\vert S^2 \vert$    \\ \cline{3-8}
 & avg  & 0.02 & 5.40 & 1.0 & 4.1 & 3.30 & 3.10 &  \\
 	\cline{1-8}
		$n$ & INS    & time & iter & $\tau$-based {\NE}? & $|\supp({\tau})|$ & $ \prod_{p=1}^m \vert \mathbb{S}^p \vert $ & & \\ 
40 & 0 & 0.26 & 18 & YES & 36 & [10, 9] &  & \\
 & 1 & 0.24 & 15 & YES & 16 & [7, 9] &  & \\
 & 2 & 0.26 & 15 & YES & 16 & [7, 9] &  & \\
 & 3 & 0.28 & 17 & YES & 22 & [10, 8] &  & \\
 & 4 & 0.04 & 5 & YES & 4 & [3, 3] &  & \\
 & 5 & 0.05 & 5 & YES & 4 & [3, 3] &  & \\
 & 6 & 0.07 & 7 & YES & 4 & [4, 4] &  & \\
 & 7 & 0.92 & 29 & YES & 36 & [17, 13] &  & \\
 & 8 & 0.20 & 15 & YES & 25 & [8, 8] &  & \\
 & 9 & 0.51 & 24 & YES & 25 & [12, 13] &  &  \\
 	&     & time    & iter      & $\tau$-based {\NE} & $|\supp({\tau})|$ & $\vert S^1 \vert$       & $\vert S^2 \vert$    \\ \cline{3-8}
 & avg  & 0.28 & 15.00 & 1.0 & 18.8 & 8.10 & 7.90 &  \\
	\cline{1-8}
		$n$ & INS    & time & iter & $\tau$-based {\NE}? & $|\supp({\tau})|$ & $ \prod_{p=1}^m \vert \mathbb{S}^p \vert $ & & \\ 
80 & 0 & 0.92 & 18 & YES & 25 & [9, 10] & & \\
 & 1 & 1.34 & 27 & YES & 64 & [14, 14] &  & \\
 & 2 & 0.68 & 16 & YES & 25 & [9, 8] &  & \\
 & 3 & 4.97 & 48 & YES & 36 & [25, 24] &  & \\
 & 4 & 4.68 & 42 & YES & 99 & [20, 23] &  & \\
 & 5 & 0.88 & 20 & YES & 49 & [11, 10] & & \\
 & 6 & 0.91 & 21 & YES & 47 & [11, 11] &  & \\
 & 7 & 0.31 & 9 & YES & 16 & [5, 5] & & \\
 & 8 & 0.71 & 18 & YES & 48 & [10, 9] &  & \\
 & 9 & 1.50 & 25 & YES & 36 & [14, 12] & & \\
 	&     & time    & iter      & $\tau$-based {\NE} & $|\supp({\tau})|$ & $\vert S^1 \vert$       & $\vert S^2 \vert$    \\ \cline{3-8}
 & avg  & 1.69 & 24.40 & 1.0 & 44.5 & 12.80 & 12.60 & \\
 	\cline{1-8}
		$n$ & INS    & time & iter & $\tau$-based {\NE}? & $|\supp({\tau})|$ & $ \prod_{p=1}^m \vert \mathbb{S}^p \vert $ & & \\ 
100 & 0 & 0.87 & 14 & YES & 25 & [8, 7] &  & \\
 & 1 & 0.89 & 15 & YES & 25 & [8, 8] & & \\
 & 2 & 5.61 & 43 & YES & 49 & [22, 22] &  & \\
 & 3 & 0.88 & 14 & YES & 25 & [8, 7] && \\
 & 4 & 1.00 & 16 & YES & 25 & [9, 8] && \\
 & 5 & 5.60 & 40 & YES & 81 & [23, 18] & & \\
 & 6 & 2.63 & 32 & YES & 79 & [17, 16] &  & \\
 & 7 & 3.69 & 27 & YES & 49 & [11, 17] &  & \\
 & 8 & 2.58 & 27 & YES & 36 & [14, 14] &  & \\
 & 9 & 8.49 & 48 & YES & 78 & [26, 23] &  & \\
  	&     & time    & iter      & $\tau$-based {\NE} & $|\supp({\tau})|$ & $\vert S^1 \vert$       & $\vert S^2 \vert$    \\ \cline{3-8}
 & avg  & 3.22 & 27.60 & 1.0 & 47.2 & 14.60 & 14.00 & \\ 
	\end{tabular}
\caption{Computational results for the determination of {\CE} on the \textbf{knapsack game} with $m=2$.}
\label{table_CE_KnapsackGame_2}
\end{table}

\bgroup
\def\arraystretch{1.5}
\begin{table}[ptbh!]
	\centering
	\tiny
	\tabcolsep=2pt
	\begin{tabular}{lr|crcccccc}
		&         & \multicolumn{7}{c}{{\SGM}}  \\[0.1cm]
		\cline{1-10}
		$n$ & INS    & time & iter & $\tau$-based {\NE}? & $|\supp({\tau})|$ & $ \prod_{p=1}^m \vert \mathbb{S}^p \vert $ & $1-\frac{Social(\sigma)}{Social(\tau)}$ & \\ 
 10 & 0 & 0.01 & 5 & YES & 4 & [3, 2, 2] & 0.00 & & \\
 & 1 & 0.03 & 7 & YES & 4 & [2, 3, 4] & 0.00 & & \\
 & 2 & 8.36 & 44 & NO & 84 & [11, 17, 18] & 0.18 & & \\
 & 3 & 0.10 & 15 & NO & 7 & [5, 6, 6] & 0.05 & & \\
 & 4 & 0.01 & 4 & YES & 1 & [2, 2, 2] & 0.00 & & \\
 & 5 & 0.02 & 6 & YES & 4 & [4, 3, 1] & 0.00 & & \\
 & 6 & 0.02 & 7 & YES & 1 & [3, 3, 3] & 0.00 & & \\
 & 7 & 0.01 & 4 & YES & 1 & [2, 2, 2] & 0.00 & & \\
 & 8 & 0.46 & 23 & NO & 16 & [8, 11, 6] & 0.10 & & \\
 & 9 & 0.04 & 10 & YES & 4 & [5, 3, 4] & 0.00 & & \\
 	&     & time    & iter      & $\tau$-based {\NE} & $|\supp({\tau})|$ & $\vert S^1 \vert$       & $\vert S^2 \vert$ &$\vert S^3 \vert$  & $1-\frac{Social(\sigma)}{Social(\tau)}$   \\ \cline{3-10}
 & avg  & 0.91 & 12.50 & 0.7 & 12.6 & 4.50 & 5.20 & 4.80 & 0.03 \\
 	\cline{1-10}
		$n$ & INS    & time & iter & $\tau$-based {\NE}? & $|\supp({\tau})|$ & $ \prod_{p=1}^m \vert \mathbb{S}^p \vert $  & $1-\frac{Social(\sigma)}{Social(\tau)}$  & \\ 
20& 0 & 0.04 & 8 & YES & 4 & [4, 2, 4] & 0.00 & & \\
 & 1 & 0.05 & 10 & YES & 9 & [4, 4, 4] & 0.00 & & \\
 & 2 & 0.01 & 4 & YES & 1 & [2, 2, 2] & 0.00 & & \\
 & 3 & 0.04 & 8 & YES & 4 & [3, 4, 3] & 0.00 & & \\
 & 4 & 0.12 & 14 & YES & 9 & [5, 6, 5] & 0.00 & & \\
 & 5 & 0.04 & 8 & YES & 1 & [3, 4, 3] & 0.00 & & \\
 & 6 & 0.81 & 26 & NO & 27 & [9, 10, 9] & 0.05 & & \\
 & 7 & 0.10 & 13 & YES & 9 & [5, 5, 5] & 0.00 & & \\
 & 8 & 0.48 & 23 & NO & 32 & [9, 8, 8] & 0.13 & & \\
 & 9 & 0.02 & 5 & YES & 1 & [3, 2, 2] & 0.00 & & \\
 	&     & time    & iter      & $\tau$-based {\NE} & $|\supp({\tau})|$ & $\vert S^1 \vert$       & $\vert S^2 \vert$  &$\vert S^3 \vert$  & $1-\frac{Social(\sigma)}{Social(\tau)}$  \\ \cline{3-10}
 & avg  & 0.17 & 11.90 & 0.8 & 9.7 & 4.70 & 4.70 & 4.50 & 0.02 \\
	\cline{1-10}
		$n$ & INS    & time & iter & $\tau$-based {\NE}? & $|\supp({\tau})|$ & $ \prod_{p=1}^m \vert \mathbb{S}^p \vert $  & $1-\frac{Social(\sigma)}{Social(\tau)}$  & \\ 
40 & 0 & 0.47 & 19 & YES & 16 & [9, 3, 9] & 0.00 & & \\
 & 1 & 83.66 & 71 & NO & 17 & [38, 16, 19] & 0.10 & & \\
 & 2 & 1.93 & 30 & NO & 36 & [10, 11, 11] & 0.03 & & \\
 & 3 & 0.27 & 14 & YES & 9 & [4, 6, 6] & 0.00 & & \\
 & 4 & tl & 134 &  & 262 & [41, 47, 48] & &  & \\
 & 5 & 1.42 & 28 & NO & 28 & [11, 8, 11] & 0.00 & & \\
 & 6 & 0.42 & 18 & YES & 10 & [7, 7, 6] & 0.00 & & \\
 & 7 & 642.17 & 102 & NO & 194 & [52, 32, 20] & 0.02 & & \\
 & 8 & 1.70 & 29 & NO & 61 & [12, 5, 14] & 0.02 & & \\
 & 9 & 1137.81 & 112 & NO & 135 & [53, 21, 40] & - & & \\
 	&     & time    & iter      & $\tau$-based {\NE} & $|\supp({\tau})|$ & $\vert S^1 \vert$       & $\vert S^2 \vert$ &$\vert S^3 \vert$  & $1-\frac{Social(\sigma)}{Social(\tau)}$   \\ \cline{3-10}
 & avg  & 207.76 & 47.00 & 0.3 & 56.22 & 21.78 & 12.11 & 15.11 & 0.02 \\
	\end{tabular}
\caption{Computational results for the determination of {\CE} on the \textbf{knapsack game} with $m=3$.}
\label{table_CE_KnapsackGame_3}
\end{table}

\subsubsection{Two-player kidney exchange game}

As in the experiments for the knapsack game, the value of $\varepsilon$ is zero. We used the instances of~\cite{constantino_new_2013}\footnote{Instances available in \url{https://rdm.inesctec.pt/dataset/ii-2019-001}} based on the US population~\citep{saidman_increasing_2006}. From this dataset, we used compatability graphs with sizes (\ie, total number of vertices) equal to 20, 40, 60 and 80 vertices. Given that the primary goal is to maximize the number of patients receiving a kidney, we used unitary weights, \ie, $w^p_c$ corresponds to the number of patients from country $p$ in the cycle $c$. For each graph size, there are 50 instances, except for size 20 where there are 49 instances. The bound on the cycles length considered was 3. Finally, we assign half of the vertices to each country.

\paragraph{{\NE} computation. } Table~\ref{table_fastNE_KEG} presents our results for the computation of Nash equilibria. Besides the previous described column entries,  the table presents the average time in seconds to determine a social optimum (``Social opt time''), the average ratio between the {\NE} social welfare and the social optimum (``price of {\NE}''),  each country average payoff decrease ratio when acting alone in comparison to joining the game and playing the {\NE} (``$\Pi^p$ decrease''), and the percentage of solved instances (``\% Solved'').

The second column of the table already reflects the difficulty of the best reaction integer programs: computing the social optimum is equivalent to optimizing the sum of the players payoffs subject to their restrictions. This problem is NP-hard when cycles are limited to length 3. Nevertheless, we were able to compute a pure {\NE} for all our instances in less than 1 second. For graphs of size 20 and 40, the loss on social welfare for playing a {\NE} is not significant, and the players' benefit for joining the game is considerable. On the other hand, for the graph of size 100, the conclusion is reversed. Thus, our results reveal the need of designing game rules ensuring that both the social welfare is increased, as well as, the players benefit for participating in the game.

\paragraph{{\CE} computation. } Table~\ref{table_fastCE_KEG} summarizes our results for the computation of {\CE}. Since the support size of all determined {\CE} is 1, it is easy to see that they are themselves {\NE}.  These results seem to indicate that it would be enough to search for {\CE} in order to determine a {\NE}. However, the determination of {\CE} is more costly as the termination step of {\SGM} must solve for each player $p$ as many Problems~\eqref{Problem:Correlated} as the size of $\mathbb{S}^p$ of the sampled game, while for the verification of {\NE}, a single best response is solved\footnote{We note that we can skip the solving of Problem~\eqref{Problem:Correlated} for the $\bar{x}^p \in \mathbb{S}^p$ such that $\sum_{x \in \mathbb{S}: x^p=\bar{x}^p}\tau(x)=0$. If we do this in the kidney exchange game, them the {\CE} computation becomes faster than the {\NE} computation since the support of the correlated equilibria is 1.}. Another advantage of the computed {\CE} is the small positive average gain between the {\NE} previously determined and the {\CE} computed. This is might be explained by the fact that the {\CE} of each sampled game optimizes social welfare.

\bgroup
\def\arraystretch{1.5}
\begin{table}[ptbh!]
	\centering
	\tiny
	\tabcolsep=2pt
	\begin{tabular}{lc|ccccccc|ccc|ccccccc|ccc}
		&   &   \multicolumn{10}{c|}{{\modSGM}}    & \multicolumn{10}{c}{{\SGM}}  \\[0.1cm] \hline 
		& & \multicolumn{7}{c|}{avg.} & & & & \multicolumn{7}{c|}{avg.}\\
		$|V|$ & Social opt & Price of  & $\Pi^A$ & $\Pi^B$    & time      & iter    & $\vert S^1 \vert$ & $\vert S^2 \vert$ &    numb. back      &pNE & \% Solved & Price of  & $\Pi^A$ & $\Pi^B$    & time      & iter  & $\vert S^1 \vert$ & $\vert S^2 \vert$ & pNE & \% Solved  \\
		& time	 &NE &  decrease &  decrease    &     &  &     & &  &        & & NE & increase & increase    &     &  & & &  & \\
  20 & 0.02 & 0.89 & 0.23 & 0.22 & 0.01 & 1.84 & 1.53 & 1.31 &  0.00 & 49 & 1.0 & 0.89 & 0.23 & 0.22 & 0.01 & 1.84 & 1.53 & 1.31 &  49 &  100 \\ 
   40 & 0.09 & 0.91 & 0.11 & 0.14 & 0.05 & 3.08 & 2.14 & 1.94 &  0.00 & 50 & 1.0 & 0.91 & 0.11 & 0.14 & 0.05 & 3.08 & 2.14 & 1.94 &  50 &  100 \\ 
 80 & 0.83 & 0.92 & 0.05 & 0.05 & 0.98 & 3.44 & 2.44 & 2.00 &  0.00 & 50 & 1.0 & 0.92 & 0.05 & 0.05 & 0.97 & 3.44 & 2.44 & 2.00 &  50 &  100 \\ 
	\end{tabular}
	\caption{Average results for the determination of (maximal) {\NE} on the  \textbf{kidney exchange game}.}
	\label{table_fastNE_KEG}
\end{table}

\bgroup
\def\arraystretch{1.5}
\begin{table}[ptbh!]
	\centering
	\tiny
	\tabcolsep=2pt
	\begin{tabular}{l|ccccccc}
		&       \multicolumn{7}{c}{{\SGM}}  \\[0.1cm] \hline 
		& \multicolumn{7}{c}{avg.} \\
	$|V|$	     & time    & iter      & $\tau$-based {\NE} & $|\supp({\tau})|$ & $\vert S^1 \vert$       & $\vert S^2 \vert$  & $1-\frac{Social(\sigma)}{Social(\tau)}$\\
 20   & 0.02 & 1.84 & 49 & 1.0 & 1.53 & 1.31 & 0.01 \\ 
 40   & 0.18 & 3.08 & 50 & 1.0 & 2.14 & 1.94 & 0.02\\
 80 &5.23 & 3.50 & 50 & 1.0 & 2.50 & 2.00 & 0.01
	\end{tabular}
	\caption{Average results for the determination of (maximal) {\CE} on the  \textbf{kidney exchange game}.}
	\label{table_fastCE_KEG}
\end{table}

\subsubsection{Competitive lot-sizing game}
Through dynamic programming, a player $p$'s best reaction~\eqref{ULSG} for a fixed $(y^{-p},x^{-p},q^{-p},h^{-p})$ can be computed in polynomial time if there are no production capacities, neither inventory costs~\citep{CarvalhoTelhaVyve}.  For this reason, we decided to concentrate on these simpler instances. In our computations, the value of $\varepsilon$ was set to  $10^{-6}$.
The parameters $a_t$, $b_t$, $F_t^p$ and $C^p_t$ were draw independently from a uniform distribution in the intervals $\left[ 20, 30 \right] \cap \mathbb{Z}$, $\left[ 1, 3 \right] \cap \mathbb{Z}$, $\left[ 10, 20 \right] \cap \mathbb{Z}$,  $\left[5, 10 \right] \cap \mathbb{Z}$, respectively.  For each value of the pair $(m,T)$, 10 instances were generated. 

For easiness of implementation and fair comparison with the computation of the potential function optimum, we do not use the dynamic programming procedure to solve a player best reaction problem, but Gurobi 9.0.0.

%

As previously mentioned, Section~\ref{sec:competitive_lotsizing_game}, the lot-sizing game is potential, which implies the existence of a pure equilibrium. In particular, each sampled game of the competitive lot-sizing game is potential and therefore, it has a pure equilibrium. In fact, our algorithms will return a pure equilibrium: both {\modSGM} and {\SGM} start with a sampled game with only one strategy for each player and thus, one pure equilibrium. This equilibrium is given to the input of our {\PNS} implementation, which implies that players' supports of size one will be prioritized leading to the computation of a pure equilibrium. This pure equilibrium will be in the input of the next {\PNS} call, resulting in a pure equilibrium output. This reasoning propagates through the algorithms' execution. Even though our algorithms find a pure equilibrium, it is expected that the potential function maximization method will provide an equilibrium faster than our methods, since our algorithms deeply depend on the initialization (which in our implementation does not take into account the players' interaction).  

Table \ref{table_fastNE_LotSizingGame} reports the results for the {\modSGM}, {\SGM} and potential function maximization. The table displays the number of periods ($T$),  the number of players ($m$) and the number of instances solved by each method (``numb. pNE'').  In this case all instances were solved within the time frame of one hour and a pure Nash equilibrium was determined by both our methods.

In this case, {\modSGM} does not present advantages with respect to {\SGM}. This is mainly due to the fact that the sampled games always have pure equilibria and our improvements have more impact when many mixed equilibria exist. The maximization of the potential functions allowed the computation of equilibria to be faster. This highlights the importance of identifying if a game is potential. On the other hand, the potential function maximization allows the determination of one equilibrium, while our method with different $Initialization$ and/or $PlayerOrder$ implementations may return different equilibria and, thus, allows larger exploration of the set of equilibria. 

Algorithm $PlayerOrder$ has a crucial impact in the number of sampled games to be explored in order to compute one equilibrium. In fact, when comparing our implementation with simply keeping the players' index order static, the impact on computational times is significant. 

We do not report our results for the computation of {\CE} since there was no social welfare improvement on the {\CE} determined, and all {\CE} computed were pure {\NE}.

\bgroup
\def\arraystretch{1.5}
\begin{table}[ptbh!]
	\centering
	\tiny
	\tabcolsep=2pt
	\begin{tabular}{lr|crcccc|c|ccrcc|c|c|c}
		&      & \multicolumn{7}{c|}{{\modSGM}}    & \multicolumn{6}{c|}{{\SGM}} & \multicolumn{2}{c}{Potential Function Maximization}  \\
		\cline{1-2} \cline{2-10} \cline{10-17}
		&      & \multicolumn{6}{c|}{avg.}  & \multicolumn{1}{c|}{numb.} &\multicolumn{5}{c|}{avg.}  & \multicolumn{1}{c|}{numb.}  &  \hspace*{0.5cm} avg \hspace*{0.5cm} & numb.  \\
		$m$ & $T$ & time & iter     &  $\vert S^1 \vert$     & $\vert S^2 \vert$ & $\vert S^3 \vert$        &\multicolumn{1}{c|}{numb. back}    & pNE      & time & iter     &  $\vert S^1 \vert$ & $\vert S^2 \vert$ & \multicolumn{1}{c|}{$\vert S^3 \vert$}  & pNE  & \hspace*{0.5cm} time \hspace*{0.5cm} & pNE   \\   \cline{3-17}  
2 & 10 & 0.08 & 15.00 & 8.00 & 8.00 & & 0.00 & 10 & 0.06 & 15.00 & 8.00 & 8.00 & & 10 & 0.00 & 10\\ 
 & 20 & 0.13 & 16.00 & 9.00 & 8.00 & & 0.00 & 10 & 0.11 & 16.00 & 9.00 & 8.00 & & 10 & 0.01 & 10\\ 
 & 50 & 0.26 & 16.10 & 9.00 & 8.10 & & 0.00 & 10 & 0.23 & 16.10 & 9.00 & 8.10 & & 10 & 0.02 & 10\\ 
  & 100 & 0.51 & 17.00 & 9.00 & 9.00 & & 0.00 & 10 & 0.44 & 17.00 & 9.00 & 9.00 & & 10 & 0.03 & 10\\
3  & 10 & 0.30 & 30.60 & 11.30 & 10.80 & 11.90  & 0.00 & 10 & 0.26 & 30.60 & 11.30 & 10.80 & 11.90 & 10 & 0.01 & 10\\ 
 & 20 & 0.42 & 31.80 & 11.90 & 11.00 & 12.10  & 0.00 & 10 & 0.35 & 31.80 & 11.90 & 11.00 & 12.10 & 10 & 0.03 & 10\\
  & 50 & 0.77 & 32.70 & 12.10 & 11.50 & 12.60  & 0.00 & 10 & 0.63 & 32.70 & 12.10 & 11.50 & 12.60 & 10 & 0.04 & 10\\  
  & 100 & 1.56 & 34.20 & 12.50 & 12.00 & 13.20  & 0.00 & 10 & 1.28 & 34.20 & 12.50 & 12.00 & 13.20 & 10 & 0.07 & 10\\ 
	\end{tabular}
	\caption{Average results for the determination of {\NE} for the \textbf{competitive lot-sizing game}.}
	\label{table_fastNE_LotSizingGame}
\end{table}
\bgroup
\def\arraystretch{1}

\subsubsection{Final remarks}

In the application of our two methods in all the studied instances of these games, backtracking never occurred. Indeed, we noticed that this is a very unlikely event (even though it may happen, as shown in Example \ref{example_backtracking}). This is the reason why both  {\modSGM} and {\SGM}, in general, coincide in the number of sampled games generated. It is in the support enumeration for each sampled game that the methods differ. The fact that in each iteration of {\modSGM}  the last added strategy is mandatory to be in the equilibrium support, makes {\modSGM} faster.  The backtracking will reveal useful for problems in which it is ``difficult'' to find the strategies of a sampled game that enable to define an equilibrium of an {\IPG}. 

At this point, for the games studied, in comparison with the number of pure profiles of strategies that may exist in a game, not too many sampled games had to be generated  in order to find an equilibrium, meaning that the challenge is to make the computation of equilibria for sampled games faster.

\paragraph{Comparison: {\modSGM} and {\PNS}.} In the case of the knapsack game, the number of strategies for each player is finite. In order to find an equilibrium of it, we can explicitly determine all feasible strategies for each player and, then apply directly {\PNS}. In Tables~\ref{table_ISM_Vs_PNS} and~\ref{table_ISM_Vs_PNS_2}, we compare this procedure with {\modSGM}, for $n=5$, $n=7$ and $n=10$ (in these cases, each player has at most $2^5=32$, $2^7=128$ and $2^{10}=1024$ feasible strategies, respectively). We note that the computational time displayed in these tables under the direct application of {\PNS} does not include the time to determine all feasible strategies for each player (although, for $n=5$, $n=7$ and $n=10$ is negligible). Based on these results it can be concluded that even for small instances,  {\modSGM} already performs better than the direct application of {\PNS},  where all strategies must  be enumerated.

\begin{table}[ptbh!]
	\centering
	\tiny
	\tabcolsep=2pt
	\begin{tabular}{llr|crccccc|crccccc}
		&	   &      & \multicolumn{7}{c|}{{\modSGM}}    & \multicolumn{6}{c}{direct {\PNS}}  \\ \hline
		$n$ & $m$ & INS    & time      & iter     &  pNE & mNE   & $ \prod_{p=1}^m \vert \mathbb{S}^p \vert $    & numb. back      & & time & pNE & mNE &  $ \prod_{p=1}^m \vert \mathbb{S}^p \vert $ \\ 	
5 & 2 & 0 & 0.01 & 3 & 1 & 0 & [2, 2] & 0 &  & 0.02 & 1 & 0 & [26, 17] & & \\ 
  &   & 1 & 0.01 & 2 & 1 & 0 & [1, 2] & 0 &  & 0.01 & 1 & 0 & [8, 14] & & \\ 
 & & 2 & 0.01 & 2 & 1 & 0 & [1, 2] & 0 &  & 0.01 & 1 & 0 & [8, 15] & & \\ 
 & & 3 & 0.00 & 1 & 1 & 0 & [1, 1] & 0 &  & 0.02 & 1 & 0 & [15, 17] & & \\ 
 & & 4 & 0.02 & 4 & 1 & 0 & [2, 3] & 0 &  & 0.01 & 1 & 0 & [17, 10] & & \\ 
 & & 5 & 0.01 & 2 & 1 & 0 & [1, 2] & 0 &  & 0.01 & 1 & 0 & [17, 15] & & \\ 
 & & 6 & 0.01 & 3 & 1 & 0 & [2, 2] & 0 &  & 0.01 & 1 & 0 & [16, 14] & & \\ 
 & & 7 & 0.01 & 2 & 1 & 0 & [1, 2] & 0 &  & 0.02 & 1 & 0 & [14, 21] & & \\ 
 & & 8 & 0.01 & 4 & 1 & 0 & [2, 3] & 0 &  & 0.01 & 1 & 0 & [20, 15] & & \\ 
 & & 9 & 0.01 & 2 & 1 & 0 & [1, 2] & 0 &  & 0.01 & 1 & 0 & [22, 9] & & \\      \cline{4-16} 
		&	   &      & \multicolumn{5}{c:}{avg.}  & \multicolumn{2}{c|}{number of} &\multicolumn{4}{c:}{avg.}  & \multicolumn{2}{c}{number of}    \\ 
		&          &     & time    & iter      & $\vert S^1 \vert$     & $\vert S^2 \vert$ & \multicolumn{1}{c:}{}   &  pNE   & mNE   & time    & $\vert S^1 \vert$     & $\vert S^2 \vert$ &  \multicolumn{1}{c:}{} & pNE   & mNE  \\   \cline{4-16} 
 & & & 0.01 & 2.50 & 1.40 & 2.10 &  & 10 & 0 & 0.01 & 16.30 & 14.70 &  & 10 & 0 \\  
		\cline{2-16}
		& $m$ & INS    & time      & iter     &  pNE & mNE   & $ \prod_{p=1}^m \vert \mathbb{S}^p \vert $    & numb. back      & & time & pNE & mNE &  $ \prod_{p=1}^m \vert \mathbb{S}^p \vert $ \\ 	
 & 3 & 0 & 0.01 & 2 & 1 & 0 & [2, 1, 1] & 0 &  & 0.03 & 1 & 0 & [14, 21, 4] & & \\ 
 & & 1 & 0.01 & 3 & 1 & 0 & [2, 1, 2] & 0 &  & 0.07 & 1 & 0 & [20, 10, 29] & & \\ 
 & & 2 & 0.02 & 4 & 0 & 1 & [2, 2, 2] & 0 &  & 1.58 & 0 & 1 & [25, 23, 16] & & \\ 
 & & 3 & 0.01 & 2 & 1 & 0 & [1, 2, 1] & 0 &  & 0.05 & 1 & 0 & [18, 18, 19] & & \\ 
 & & 4 & 0.02 & 4 & 1 & 0 & [2, 2, 2] & 0 &  & 0.07 & 1 & 0 & [12, 21, 17] & & \\ 
 & & 5 & 0.02 & 4 & 1 & 0 & [2, 2, 2] & 0 &  & 0.12 & 1 & 0 & [20, 19, 16] & & \\ 
 & & 6 & 0.01 & 3 & 1 & 0 & [2, 2, 1] & 0 &  & 0.09 & 1 & 0 & [14, 18, 15] & & \\ 
 & & 7 & 0.02 & 4 & 1 & 0 & [2, 2, 2] & 0 &  & 0.05 & 1 & 0 & [11, 22, 14] & & \\ 
 & & 8 & 0.02 & 5 & 1 & 0 & [3, 2, 2] & 0 &  & 0.03 & 1 & 0 & [8, 20, 16] & & \\ 
 & & 9 & 0.03 & 5 & 0 & 1 & [2, 3, 2] & 0 &  & 0.93 & 0 & 1 & [14, 21, 22] & & \\    \cline{4-16} 
		&	   &      & \multicolumn{5}{c:}{avg.}  & \multicolumn{2}{c|}{number of} &\multicolumn{4}{c:}{avg.}  & \multicolumn{2}{c}{number of}    \\ 
		&          &     & time    & iter      & $\vert S^1 \vert$     & $\vert S^2 \vert$ & \multicolumn{1}{c:}{ $\vert S^3 \vert$}   &  pNE   & mNE   & time    & $\vert S^1 \vert$     & $\vert S^2 \vert$ &  \multicolumn{1}{c:}{$\vert S^3 \vert$}  & pNE   & mNE  \\   \cline{4-16} 
 & & & 0.02 & 3.60 & 2.00 & 1.90 & 1.70 & 8 & 2 & 0.30 & 15.60 & 19.30 & 16.80 & 8 & 2 \\ 
		\cline{1-2} \cline{2-10} \cline{10-16} 
		$n$ & $m$ & INS    & time      & iter     &  pNE & mNE   & $ \prod_{p=1}^m \vert \mathbb{S}^p \vert $    & numb. back      & & time & pNE & mNE &  $ \prod_{p=1}^m \vert \mathbb{S}^p \vert $ \\ 	
7 & 2 & 0 & 0.02 & 5 & 0 & 1 & [3, 3] & 0 &  & 10.28 & 0 & 1 & [28, 79] & & \\ 
 & & 1 & 0.01 & 3 & 1 & 0 & [2, 2] & 0 &  & 0.09 & 1 & 0 & [58, 32] & & \\ 
 & & 2 & 0.00 & 1 & 1 & 0 & [1, 1] & 0 &  & 0.16 & 1 & 0 & [93, 51] & & \\ 
 & & 3 & 0.01 & 2 & 1 & 0 & [2, 1] & 0 &  & 0.12 & 1 & 0 & [55, 50] & & \\ 
 & & 4 & 0.01 & 3 & 0 & 1 & [2, 2] & 0 &  & 6.62 & 0 & 1 & [50, 43] & & \\ 
 & & 5 & 0.03 & 6 & 0 & 1 & [4, 3] & 0 &  & 1739.92 & 0 & 1 & [68, 61] & & \\ 
 & & 6 & 0.01 & 3 & 1 & 0 & [2, 2] & 0 &  & 0.33 & 1 & 0 & [70, 59] & & \\ 
 & & 7 & 0.01 & 2 & 1 & 0 & [2, 1] & 0 &  & 0.28 & 1 & 0 & [82, 67] & & \\ 
 & & 8 & 0.01 & 4 & 0 & 1 & [3, 2] & 0 &  & 5.51 & 0 & 1 & [53, 36] & & \\ 
 & & 9 & 0.02 & 5 & 0 & 1 & [3, 3] & 0 &  & 13.89 & 0 & 1 & [45, 65] & & \\  \cline{4-16} 
		&   &      & \multicolumn{5}{c:}{avg.}  & \multicolumn{2}{c|}{number of} &\multicolumn{4}{c:}{avg.}  & \multicolumn{2}{c}{number of}    \\ 
		&       &     & time    & iter      & $\vert S^1 \vert$     & $\vert S^2 \vert$ &  \multicolumn{1}{c:}{} &  pNE   & mNE   & time    & $\vert S^1 \vert$     & $\vert S^2 \vert$ & \multicolumn{1}{c:}{}  & pNE   & mNE  \\   \cline{4-16} 
 & & & 0.01 & 3.40 & 2.40 & 2.00 &  & 5 & 5 & 177.72 & 60.20 & 54.30 &  & 5 & 5 \\ 
		\cline{2-16}
		& $m$ & INS    & time      & iter     &  pNE & mNE   & $ \prod_{p=1}^m \vert \mathbb{S}^p \vert $    & numb. back      & & time & pNE & mNE &  $ \prod_{p=1}^m \vert \mathbb{S}^p \vert $ \\ 	
 & 3 & 0 & 0.02 & 4 & 1 & 0 & [2, 2, 2] & 0 &  & 2.79 & 1 & 0 & [122, 80, 118] & & \\
 & & 1 & 0.01 & 2 & 1 & 0 & [2, 1, 1] & 0 &  & 0.38 & 1 & 0 & [33, 20, 111] & & \\ 
 & & 2 & 0.02 & 4 & 0 & 1 & [3, 1, 2] & 0 &  & 6.36 & 0 & 1 & [64, 7, 20] & & \\ 
 & & 3 & 0.02 & 4 & 1 & 0 & [2, 2, 2] & 0 &  & 1.79 & 1 & 0 & [30, 52, 74] & & \\ 
 & & 4 & 0.05 & 8 & 0 & 1 & [3, 3, 4] & 0 &  & 6.90 & 0 & 1 & [71, 61, 56] & & \\ 
 & & 5 & 0.03 & 5 & 0 & 1 & [2, 3, 2] & 0 &  & 1.46 & 0 & 1 & [59, 66, 74] & & \\ 
 & & 6 & 0.02 & 4 & 1 & 0 & [2, 2, 2] & 0 &  & 1.22 & 1 & 0 & [70, 64, 66] & & \\ 
 & & 7 & 0.01 & 3 & 1 & 0 & [2, 2, 1] & 0 &  & 2.77 & 1 & 0 & [83, 91, 85] & & \\ 
 & & 8 & 0.06 & 7 & 0 & 1 & [4, 2, 3] & 0 &  & tl & 0  &  0 [75, 88, 82] & \\ 
 & & 9 & 0.03 & 6 & 0 & 1 & [3, 2, 3] & 0 &  & 576.75 & 0 & 1 & [106, 73, 91] & & \\ \cline{4-16} 
		&   &      & \multicolumn{5}{c:}{avg.}  & \multicolumn{2}{c|}{number of} &\multicolumn{4}{c:}{avg.}  & \multicolumn{2}{c}{number of}    \\ 
		&      &     & time    & iter      & $\vert S^1 \vert$     & $\vert S^2 \vert$ &  \multicolumn{1}{c:}{$\vert S^3 \vert$}   &  pNE   & mNE   & time    & $\vert S^1 \vert$     & $\vert S^2 \vert$ &  \multicolumn{1}{c:}{$\vert S^3 \vert$}  & pNE   & mNE  \\   \cline{4-16} 
 & & & 0.03 & 4.70 & 2.50 & 2.00 & 2.20 & 5 & 5 & 66.71 & 70.89 & 57.11 & 77.22 & 5 & 4 
	\end{tabular}
	\caption{Computational results for the {\modSGM} and {\PNS} to the knapsack game with $n=5,7$.}
	\label{table_ISM_Vs_PNS}
\end{table}

\begin{table}[ptbh!]
	\centering
	\tiny
	\tabcolsep=2pt
	\begin{tabular}{llr|crccccc|crccccc}
		&		   &      & \multicolumn{7}{c|}{{\modSGM}}    & \multicolumn{6}{c}{direct {\PNS}}  \\ \hline
		$n$ & $m$ & INS    & time      & iter     &  pNE & mNE   & $ \prod_{p=1}^m \vert \mathbb{S}^p \vert $    & numb. back      & & time & pNE & mNE &  $ \prod_{p=1}^m \vert \mathbb{S}^p \vert $ \\ 	
10 & 2 & 0 & 0.01 & 3 & 1 & 0 & [2, 2] & 0 &  & 57.00 & 1 & 0 & [879, 540] & & \\ 
 & & 1 & 0.02 & 4 & 1 & 0 & [3, 2] & 0 &  & 3.73 & 1 & 0 & [734, 73] & & \\ 
 & & 2 & 0.01 & 2 & 1 & 0 & [2, 1] & 0 &  & 1.29 & 1 & 0 & [587, 69] & & \\ 
 & & 3 & 0.02 & 4 & 0 & 1 & [3, 2] & 0 &  & tl & 0  &  0 & [498, 493] & \\ 
 & & 4 & 0.01 & 2 & 1 & 0 & [1, 2] & 0 &  & 21.88 & 1 & 0 & [529, 453] & & \\ 
 & & 5 & 0.09 & 10 & 0 & 1 & [6, 5] & 0 &  & tl & 0  &  0 & [468, 567] & \\ 
 & & 6 & 0.01 & 3 & 1 & 0 & [2, 2] & 0 &  & 88.56 & 1 & 0 & [555, 501] & & \\ 
 & & 7 & 0.01 & 3 & 1 & 0 & [2, 2] & 0 &  & 12.23 & 1 & 0 & [522, 532] & & \\ 
 & & 8 & 0.01 & 2 & 1 & 0 & [2, 1] & 0 &  & 684.92 & 1 & 0 & [340, 724] & & \\ 
 & & 9 & 0.01 & 2 & 1 & 0 & [2, 1] & 0 &  & 112.77 & 1 & 0 & [454, 544] & & \\    \cline{4-16} 
		&   &      & \multicolumn{5}{c:}{avg.}  & \multicolumn{2}{c|}{number of} &\multicolumn{4}{c:}{avg.}  & \multicolumn{2}{c}{number of}    \\ 
		&      &     & time    & iter      & $\vert S^1 \vert$     & $\vert S^2 \vert$ &  \multicolumn{1}{c:}{}  &  pNE   & mNE   & time       & $\vert S^1 \vert$     & $\vert S^2 \vert$ &  \multicolumn{1}{c:}{}   & pNE   & mNE  \\   \cline{4-16} 
 & & & 0.02 & 3.50 & 2.50 & 2.00 &  & 8 & 2 & 122.80 & 575.00 & 429.50 &  & 8 & 0  \\         
		\cline{2-16}
		& $m$ & INS    & time      & iter     &  pNE & mNE   & $ \prod_{p=1}^m \vert \mathbb{S}^p \vert $    & numb. back      & & time & pNE & mNE &  $ \prod_{p=1}^m \vert \mathbb{S}^p \vert $ \\ 
 & 3 & 0 & 0.06 & 9 & 0 & 1 & [4, 3, 4] & 0 &  & tl & 0  &  0 & [464, 609, 236] & \\ 
 & & 1 & 0.05 & 7 & 0 & 1 & [2, 3, 4] & 0 &  & tl & 0  &  0 & [618, 263, 224] & \\ 
 & & 2 & 0.02 & 4 & 1 & 0 & [2, 2, 2] & 0 &  & 152.44 & 1 & 0 & [756, 200, 457] & & \\ 
 & & 3 & 0.04 & 7 & 1 & 0 & [3, 3, 3] & 0 &  & tl & 0  &  0 & [528, 548, 691] & \\ 
 & & 4 & 0.06 & 8 & 1 & 0 & [3, 4, 3] & 0 &  & 1059.49 & 1 & 0 & [418, 587, 700] & & \\ 
 & & 5 & 0.09 & 9 & 0 & 1 & [4, 4, 3] & 0 &  & tl & 0  &  0 & [509, 491, 540] & \\ 
 & & 6 & 0.04 & 7 & 1 & 0 & [3, 3, 3] & 0 &  & 3042.26 & 1 & 0 & [445, 491, 482] & & \\ 
 & & 7 & 0.03 & 5 & 0 & 1 & [3, 2, 2] & 0 &  & 1589.25 & 0 & 1 & [634, 525, 715] & & \\ 
 & & 8 & 0.03 & 6 & 1 & 0 & [3, 3, 2] & 0 &  & 1408.38 & 1 & 0 & [500, 574, 436] & & \\ 
 & & 9 & 0.04 & 7 & 0 & 1 & [3, 3, 3] & 0 &  & tl & 0  &  0 & [559, 466, 303] & \\   \cline{4-16} 
		&  &      & \multicolumn{5}{c:}{avg.}  & \multicolumn{2}{c|}{number of} &\multicolumn{4}{c:}{avg.}  & \multicolumn{2}{c}{number of}    \\ 
		&     &     & time    & iter      & $\vert S^1 \vert$     & $\vert S^2 \vert$ &   \multicolumn{1}{c:}{$\vert S^3 \vert$}   &  pNE   & mNE   & time    & $\vert S^1 \vert$     & $\vert S^2 \vert$ &   \multicolumn{1}{c:}{$\vert S^3 \vert$}  & pNE   & mNE  \\   \cline{4-16} 
 & & & 0.04 & 6.90 & 3.00 & 3.00 & 2.90 & 5 & 5 &1450.36 & 550.60 & 475.40 & 558.00 & 4 & 1 
	\end{tabular}
	\caption{Computational results for the {\modSGM} and {\PNS} to the knapsack game with $n=10$.}
	\label{table_ISM_Vs_PNS_2}
\end{table}

\section{Conclusions and further directions}\label{sec:conclusion}

We showed that the problem of equilibria existence for {\IPG}s is $\Sigma_2^p$-complete and, even if an equilibrium exists, its computation is at least PPAD-hard.  This is not surprising, since verifying if a profile of strategies is an equilibrium for an {\IPG} implies solving each player's best response optimization, which can be an NP-complete problem. Thus, the goal of this paper was to also contribute with an algorithmic approach for the computation of equilibria with a reasonable running time in practice.

To the best of our knowledge, the novel framework proposed and evaluated in this paper is the first  addressing the computation of equilibria for non-cooperative simultaneous games in which the players' goal in the game is modelled through a mixed integer program. These games are of practical interest given that in many real-world applications, players' decision problem are combinatorial and interact.

In this work, we combined algorithms (and tools) from mathematical programming and game theory to devise a novel method to determine Nash equilibria. Our basic method, {\SGM}, iteratively determines equilibria of normal-form games which progressively improve the approximation to the original {\IPG}.  In order to make the algorithm faster in practice, special features were added. For this purpose, we devised the modified {\SGM}. We also discussed our methodology extension to correlated equilibria. Our algorithms were experimentally validated through three particular games:  the knapsack, the kidney exchange and the competitive lot-sizing games. For the knapsack game, our methods provide equilibria to medium size instances within the time frame of one hour. The results show that this is a hard game which is likely to have strictly mixed equilibria (\ie, no pure strategy is played with probability 1). The hardness comes from the conflicts that projects selected by different players have in their payoffs: for some projects $i$ a player $p$ can benefit from player $k$ simultaneous investment, while player $k$ is penalized. Surprisingly, {\CE} which are much faster to determine can help to find {\NE}. For the kidney exchange and the competitive lot-sizing game, our approaches could efficiently determine a pure equilibrium. However, it remains as a challenge the understanding on how our method initialization can result in different equilibria. Characterizing the set of equilibria is crucial to understand the game properties, specially, in terms of social welfare: if all equilibria are far from the social optimum (the so called \emph{price of stability}), policy makers should consider the re-design of the game rules.

Note that for the instances solved by our algorithms, there is an exponential (knapsack and kidney exchange games) or infinite ({competitive lot-sizing game) number of pure profiles of strategies. However, by observing the computational results, a small number of explicitly enumerated pure strategies was enough to find an equilibrium. For this reason, the explicitly enumerated strategies (the sampled games) are usually ``far'' from describing (even partially) a player $p$'s polytope $\conv(X^p)$ and thus, at this point, this information is not used in {\PNS} to speed up its computations. For instance, Theorem~\ref{lemma_finitelysupported} and Lemma~\ref{lemma:support_size_knapsack_game} did not reduce the number of supports enumerated by {\PNS} in each iteration of {\modSGM}. Due to the fact that it is in {\PNS} that our algorithms struggle the most, its improvement is the first aspect to further study; we believe that exploring the possibility of extracting information from each player's polytope of feasible strategies will be the crucial ingredient for this.

	There is a set of natural questions that this work opens. Can we adapt {\modSGM} to compute all equilibria (or characterize the set of equilibria)? Can we compute an equilibrium satisfying a specific property (e.g., computing the equilibrium that maximizes the social welfare, computing a non-dominated equilibrium)? Will in practice players play equilibria that are ``hard" to find? If a game has multiple equilibria, how to decide among them?
	From a mathematical point of view, the first two questions embody a big challenge, since there seems to be hard to extract problem structure to the general {\IPG} class of games. The two last questions raise another one, which is the possibility of considering different solution concepts to {\IPG}s.

\section*{Acknowledgements}

The authors wish to thank  Sriram Sankaranarayanan for the multiple questions on a previous version of the paper that lead to  significant improvement of this work exposition.


\bibliographystyle{spbasic} 
\bibliography{ComputingNEReferences}

\newpage

\begin{appendices}

\section{Illustration of backtracking step} \label{app:example_backtracking}
\begin{example}
	Consider the two-player knapsack game described by the following optimization problems
	\begin{subequations}
		\begin{alignat*}{4}
			Player A: \ \ &   \max_{x^A \in \lbrace0,1\rbrace^n}  && 15x^A_1+8x^A_2-3x^A_3+43x^A_4-15x^ A_5+39x^ A_1x^B_1-90x^A_2x^B_2\\
			&                                       & &     +11x^A_3x^B_3-84x^A_4x^B_4-43x^A_5x^B_5\\[0.4ex]
			&\mbox{subject to~~~}  && 70x^A_1-79x^A_2-8x^A_3-62x^A_4-96x^ A_5 \leq -140
		\end{alignat*}
	\end{subequations}
	\begin{subequations}
		\begin{alignat*}{4}
			Player B: \ \ &   \max_{x^B \in \lbrace0,1\rbrace^n}  && 24x^B_1+13x^B_2+44x^B_3-x^A_4-45x^ B_5-73x^ A_1x^B_1-58x^A_2x^B_2\\
			&                                         & & -78x^A_3x^B_3-49x^A_4x^B_4+72x^A_5x^B_5\\[0.4ex]
			&\mbox{subject to~~~}  && 69x^B_1+25x^B_2-39x^B_3-74x^B_4+70x^ B_5 \leq 40.8
		\end{alignat*}
	\end{subequations}
	
	In what follows, we go through each sampled game generated by  {\modSGM}. Figure \ref{Fig:example_backtracking} displays the sampled games using a bimatrix-form representation.
	
	\paragraph{Sampled game 0. } The {\NE} is $\sigma_0=(1;1)$. However, in the original game, player $A$ has incentive to deviate to $x(1)=(0,0,1,1,1)$. 
	
	\paragraph{Sampled game 1. } The {\NE}  is $\sigma_1=(0,1;1)$. However, in the original game, player $B$ has incentive to deviate to $x(2)=(0,1,0,0,0)$. 
	
	\paragraph{Sampled game 2. } The {\NE} is $\sigma_2=(0,1;0,1)$. However, player $A$ has incentive to deviate to $x(3)=(0,0,0,1,1)$.
	
	\paragraph{Sampled game 3. }  The  {\NE} is mixed with $\supp( \sigma_3^A)= \lbrace  (0,0,1,1,1),(0,0,0,1,1) \rbrace$ and $\supp( \sigma_3^B)= \lbrace (1,1,1,1,0) , (0,1,0,0,0) \rbrace$, $\sigma_3=(0,\frac{3}{13},\frac{10}{13};\frac{3}{11},\frac{8}{11})$. However, in the original game, player $B$ has incentive to deviate to $x(4)=(0,0,1,0,1)$. 
	
	\paragraph{Sampled game 4. }  The  {\NE} is $\sigma_4=(1,0,0;0,0,1)$. However, in the original game, player $A$ has incentive to deviate to $x(5)=(0,1,1,1,0)$. 
	
	\paragraph{Sampled game 5. }  There is no  {\NE} with $x(5)=(0,1,1,1,0)$ in the support of player $A$.  Thus, initialize backtracking. 
	
	\begin{figure}[h]\center \tiny
		\begin{tabular}{cc|ccc}
			\multicolumn{5}{c}{  Sampled game 0 } \\[0.2cm]
			& & \multicolumn{1}{c}{ Player B }    \\ 
			& & (1,1,1,1,0)\\ 
			\cline{2-3}
			Player A & (1,1,0,1,1)  & (-84,-100) \\ \\
		\end{tabular}
		\begin{tabular}{cc|ccc}
			\multicolumn{5}{c}{ Sampled game 1  }\\[0.2cm]
			& & \multicolumn{1}{c}{ Player B}     \\ 
			& & (1,1,1,1,0) \\ \cline{2-3}
			\multirow{ 2}{*}{Player A}  & (1,1,0,1,1) &  (-84,-100) \\ 
			&  \cellcolor{Cyan} (0,0,1,1,1)& (-48,-47)    \\ \\
		\end{tabular} 
		\begin{tabular}{cc|ccc}
			\multicolumn{5}{c}{Sampled game 2 } \\[0.2cm]
			& & \multicolumn{2}{c}{Player B }    \\ 
			& & (1,1,1,1,0) & \cellcolor{Cyan}   (0,1,0,0,0)\\ \cline{2-4}
			\multirow{ 2}{*}{Player A} & (1,1,0,1,1) & (-84,-100) & (-39,-45) \\ 
			&  (0,0,1,1,1) &(-48,-47) &(25,13)    \\ \\
		\end{tabular} 
		\begin{tabular}{cc|ccc}
			\multicolumn{5}{c}{Sampled game 3 } \\[0.2cm]
			& & \multicolumn{2}{c}{ Player B }    \\ 
			& &  (1,1,1,1,0) & (0,1,0,0,0) \\ \cline{2-4}
			&(1,1,0,1,1) & (-84,-100) & (-39,-45)  \\ 
			Player A  &  (0,0,1,1,1) &(-48,-47) &(25,13)        \\  
			&  \cellcolor{Cyan} (0,0,0,1,1) & (-56,31) & (28,13)\\ \\
		\end{tabular} 
		\begin{tabular}{cc|ccc}
			\multicolumn{5}{c}{ Sampled game 4 }\\[0.2cm]
			& & \multicolumn{3}{c}{ Player B}     \\ 
			& &  (1,1,1,1,0) & (0,1,0,0,0) & \cellcolor{Cyan}  (0,0,1,0,1) \\ \cline{2-5}
			&(1,1,0,1,1) & (-84,-100) & (-39,-45)  & (8,71)\\ 
			Player A  &  (0,0,1,1,1) &(-48,-47) &(25,13)  & (-7,-7)      \\  
			& (0,0,0,1,1) & (-56,31) & (28,13) & (-15,71)\\ \\
		\end{tabular} 
		\begin{tabular}{cc|ccc}
			\multicolumn{5}{c}{ Sampled game 5  }\\[0.2cm]
			& & \multicolumn{3}{c}{ Player B}     \\ 
			& &  (1,1,1,1,0) & (0,1,0,0,0) & (0,0,1,0,1) \\ \cline{2-5}
			&(1,1,0,1,1) & (-84,-100) & (-39,-45)  & (8,71)\\ 
			Player A  &   (0,0,1,1,1) &(-48,-47) &(25,13)  & (-7,-7)      \\  
			& (0,0,0,1,1) & (-56,31) & (28,13) & (-15,71)\\
			&  \cellcolor{Cyan} (0,1,1,1,0) & (-115,-105) & (-42,22) & (59,2)\\
			
		\end{tabular} 
		\begin{tabular}{cc|ccc}
			\multicolumn{5}{c}{ Revisiting Sampled game 4  }\\[0.2cm]
			& & \multicolumn{2}{c}{ Player B}     \\ 
			& &  (1,1,1,1,0) & (0,1,0,0,0) & \cellcolor{Cyan} (0,0,1,0,1) \\ \cline{2-5}
			&(1,1,0,1,1) & (-84,-100) & (-39,-45)  & (8,71)\\ 
			Player A  &   (0,0,1,1,1) &(-48,-47) &(25,13)  & (-7,-7)      \\  
			& (0,0,0,1,1) & (-56,31) & (28,13) & (-15,71)\\
			& \cellcolor{Gray_0} (0,1,1,1,0) & (-115,-105) & (-42,22) & (59,2)
		\end{tabular} 
		\caption{Modified {\SGM} applied to Example \ref{example_backtracking}. The strategies in \colorbox{Cyan}{cyan} must be in the equilibrium support, while the strategies in \colorbox{Gray_0}{gray} are not considered in the support enumeration.}
		\label{Fig:example_backtracking}
	\end{figure}
	
	\paragraph{Revisiting sampled game 4. } Keep the best reaction strategy $x^A=(0,1,1,1,0)$ that originated the sampled game 5, but do not consider it in the support enumeration (this strategy only appears in the Feasibility Problem in order to avoid the repetition of equilibria). A {\NE} with $x^B=(0,0,1,0,1)$ in the support is computed: $\sigma_4=(0,\frac{29}{39},\frac{10}{39},0;0,\frac{8}{11},\frac{3}{11})$  with supports $\supp( \sigma_4^A)= \lbrace   (0,0,1,1,1),(0,0,0,1,1) \rbrace$ and $\supp( \sigma_4^B)= \lbrace (0,1,0,0,0) , (0,0,1,0,1) \rbrace$. This {\NE} is a {\NE} of the original game.
	\label{example_backtracking}
\end{example}

\section{Kidney exchange game with $L=3$}\label{app:no_potential}

\citet{Carvalho2016} claimed that when $w_c^p=|\{v \in V^p: v \in c \}|$, \ie, number of pairs in the cycle, 
 \begin{equation}
 \Phi(y^A,y^B) = \sum_{c \in C^A} w^A_c y_c^A+ \sum_{c \in C^B} w^B_c y_c^B + \sum_{c \in I:w_c^A=w_c^B=1} y_c + \frac{3}{2} \sum_{c \in I: w_c^A= 2 \vee w_c^B=2} y_c(y^A,y^B)
 \label{eq:conjecture_potential}
 \end{equation}
where $y(y^A,y^B)$ solves~\eqref{NKEG_IA}, is a (\emph{non-exact}) potential function of their game. This is false as the following example shows:
\begin{example}
Consider the instance of Figure~\ref{fig:Example1}. The green strategy (upper figure) of country A leads it to a payoff of 8 (note that no international exchanges are available), while by unilaterally deviating to the blue strategy (lower figure), country A gets a payoff of 9 (note that in blue it is implicit the independent agent optimal international selection $y$). In the upper figure the value of function~\eqref{eq:conjecture_potential} is 8, while  in the lower figure the value is 7.5. This shows that function~\eqref{eq:conjecture_potential} is not potential.

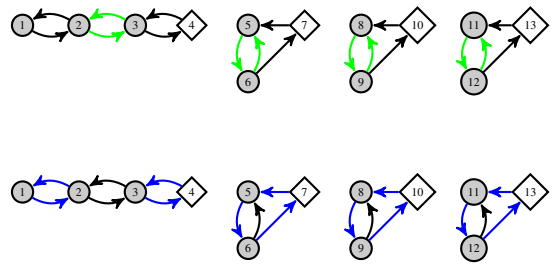
\begin{figure}[h!]
  \centering
\begin{tikzpicture}[-,>=stealth',shorten >=0.3pt,auto,node distance=1cm,
  thick,countryA node/.style={circle,fill=black!20,draw},countryB node/.style={diamond,draw}, scale=.75, transform shape]]
\tiny
  \tikzstyle{matched} = [draw,line width=3pt,-]

  \node[countryA node] (1) {1};
  \node[countryA node] (2) [right of=1] {2};
  \node[countryA node] (3) [right of=2] {3};
  \node[countryB node] (4) [right of=3] {4};
   \node[countryA node] (5) [right of=4] {5};
   \node[countryA node] (6) [below of=5] {6};
   \node[countryB node] (7) [right of=5] {7};
    \node[countryA node] (8) [right of=7] {8};
   \node[countryA node] (9) [below of=8] {9};
   \node[countryB node] (10) [right of=8] {10};
      \node[countryA node] (11) [right of=10] {11};
   \node[countryA node] (12) [below of=11] {12};
   \node[countryB node] (13) [right of=11] {13};
 
  \path[->,every node/.style={font=\sffamily\small}]
    ( 1) edge[bend right] node {} (2)
     ( 2) edge[bend right] node {} (1)
    ( 2) edge[bend right,green] node {} (3)
    ( 3) edge[bend right,green] node {} (2)
   ( 3) edge[bend right] node {} (4)
   (4) edge[bend right] node {} (3)
   (5) edge[bend right,green] node {} (6)
   (6) edge[bend right,green] node {} (5)
   (7) edge node {} (5)
   (6) edge node {} (7)
   (8) edge[bend right,green] node {} (9)
   (9) edge[bend right,green] node {} (8)
   (10) edge node {} (8)
   (9) edge node {} (10)
     (11) edge[bend right,green] node {} (12)
   (12) edge[bend right,green] node {} (11)
   (13) edge node {} (11)
   (12) edge node {} (13);      
\end{tikzpicture}

\vspace{1cm}

\begin{tikzpicture}[-,>=stealth',shorten >=0.3pt,auto,node distance=1cm,
  thick,countryA node/.style={circle,fill=black!20,draw},countryB node/.style={diamond,draw}, scale=.75, transform shape]]
\tiny
  \tikzstyle{matched} = [draw,line width=3pt,-]

  \node[countryA node] (1) {1};
  \node[countryA node] (2) [right of=1] {2};
  \node[countryA node] (3) [right of=2] {3};
  \node[countryB node] (4) [right of=3] {4};
   \node[countryA node] (5) [right of=4] {5};
   \node[countryA node] (6) [below of=5] {6};
   \node[countryB node] (7) [right of=5] {7};
    \node[countryA node] (8) [right of=7] {8};
   \node[countryA node] (9) [below of=8] {9};
   \node[countryB node] (10) [right of=8] {10};
      \node[countryA node] (11) [right of=10] {11};
   \node[countryA node] (12) [below of=11] {12};
   \node[countryB node] (13) [right of=11] {13};
 
  \path[->,every node/.style={font=\sffamily\small}]
    ( 1) edge[bend right,blue] node {} (2)
     ( 2) edge[bend right,blue] node {} (1)
    ( 2) edge[bend right] node {} (3)
    ( 3) edge[bend right] node {} (2)
   ( 3) edge[bend right,blue] node {} (4)
   (4) edge[bend right,blue] node {} (3)
   (5) edge[bend right,blue] node {} (6)
   (6) edge[bend right] node {} (5)
   (7) edge[blue] node {} (5)
   (6) edge[blue] node {} (7)
   (8) edge[bend right,blue] node {} (9)
   (9) edge[bend right] node {} (8)
   (10) edge[blue] node {} (8)
   (9) edge[blue] node {} (10)
     (11) edge[bend right,blue] node {} (12)
   (12) edge[bend right] node {} (11)
   (13) edge[blue] node {} (11)
   (12) edge[blue] node {} (13);      
\end{tikzpicture}
\caption{A kidney exchange compatability graph instance with cycles length bounded by 3, where the circle vertices belong to country A and the diamond vertices belong to country B.}
\label{fig:Example1}
\end{figure}

\end{example}

\end{appendices}

\end{document}